\documentclass[a4paper,11pt]{article}

\usepackage{graphicx}
\usepackage{latexsym}
\usepackage{amsmath,amssymb,mathtools}
\usepackage[ruled,vlined,linesnumbered]{algorithm2e}
\SetKwInOut{Input}{Input}\SetKwInOut{Output}{Output}
\usepackage[numbers]{natbib}  
\usepackage{hyperref,url}
\usepackage{doi}
\usepackage{bm}
\usepackage{tikz}
\usetikzlibrary{arrows.meta,positioning,math,calc}

\newcommand{\algref}[1]{Algorithm~\ref{alg:#1}}
\renewcommand{\eqref}[1]{(\ref{eq:#1})}
\newcommand{\secref}[1]{Section~\ref{sec:#1}}

\newcommand{\figref}[1]{Figure~\ref{fig:#1}}
\newcommand{\corref}[1]{Corollary~\ref{cor:#1}}
\newcommand{\lemref}[1]{Lemma~\ref{lem:#1}}
\newcommand{\lineref}[1]{Line~\ref{line:#1}}

\newcommand{\thmref}[1]{Theorem~\ref{thm:#1}}

\newcommand{\MA}{\mathcal{A}}
\newcommand{\MB}{\mathcal{B}}

\newcommand{\MF}{\mathcal{F}}
\newcommand{\MG}{\mathcal{G}}
\newcommand{\MI}{\mathcal{I}}

\newcommand{\MN}{\mathcal{N}}
\newcommand{\MP}{\mathcal{P}}

\newcommand{\MR}{\mathcal{R}}
\newcommand{\MS}{\mathcal{S}}

\newcommand{\MY}{\mathcal{Y}}
\newcommand{\MZ}{\mathcal{Z}}

\newcommand{\bbA}{\mathbb{A}}
\newcommand{\bbB}{\mathbb{B}}

\newcommand{\bbR}{\mathbb{R}}

\newcommand{\bbZ}{\mathbb{Z}}

\newcommand{\cl}{{\mathit{cl}}}
\newcommand{\propag}{\mathrm{pr}}

\newcommand{\Tmin}{\tau_{\mathrm{min}}}
\newcommand{\Tmem}{\tau_{\mathrm{mem}}}
\newcommand{\Tclo}{\tau_{\mathrm{clo}}}
\newcommand{\Tcol}{\tau_{\mathrm{col}}}

\newcommand{\deggeq}[1]{\mbox{deg$\geq$#1}\xspace}
\newcommand{\outdeggeq}[1]{\mbox{deg$^+$$\geq$#1}\xspace}
\newcommand{\indeggeq}[1]{\mbox{deg$^-$$\geq$#1}\xspace}
\newcommand{\true}{\ifmmode\text{\normalfont\scshape True}\else{\normalfont\scshape True}\fi\xspace}
\newcommand{\false}{\ifmmode\text{\normalfont\scshape False}\else{\normalfont\scshape False}\fi\xspace}
\newcommand{\wasInTree}{\ifmmode\text{\normalfont\scshape in-Tree}\else{\normalfont\scshape in-Tree}\fi\xspace}

\newcommand{\probMRS}{{\sc MinRSs of $k$-Core}}
\renewcommand{\widehat}[1]{\cup#1}
\newtheorem{cor}{Corollary}
\newtheorem{lem}{Lemma}

\newtheorem{thm}{Theorem}

\newenvironment{proof}{\medskip
  \noindent{\scshape Proof:}}{\quad $\Box$\medskip}
\newcommand{\myqed}{\quad $\Box$\medskip}

\long\def\invis#1{}

\title{Fast Enumeration of Minimal Removable Sets in Monotone Systems with Application to Core Collapse Analysis\thanks{This work is partially supported by JSPS KAKENHI 25K14993.}} 
\author{Kan Shota\thanks{Graduate School of Informatics, Kyoto University, \href{mailto:shota.kan@amp.i.kyoto-u.ac.jp}{\tt shota.kan@amp.i.kyoto-u.ac.jp}}\and Kazuya Haraguchi\thanks{Department of Logistics and Information Engineering, Tokyo University of Marine Science and Technology, \href{mailto:haraguchi@kaiyodai.ac.jp}{\tt haraguchi@kaiyodai.ac.jp}}}
\date{}

\begin{document}
\maketitle

\begin{abstract}
  In network vulnerability analysis,
  it is crucial to evaluate the robustness of $k$-cores against vertex removals.
  A $k$-core is often fragile since removing a few vertices can
  trigger a large reduction in the core size, a phenomenon known as core collapse.
  In this paper, we study the problem of enumerating all minimal removable sets (MinRSs) of a given $k$-core,
  where a MinRS is a minimal nonempty set of vertices whose removal
  results in a smaller $k$-core graph.
  We consider this problem within a general mathematical framework
  based on monotone systems. 
  We show that, for a monotone system that is given with an underlying graph $G=(V,E)$,
  all MinRSs of a solution can be enumerated in $O((n+m)n\tau_\omega)$ time, where
  $n=|V|$, $m=|E|$ and $\tau_\omega$ denotes the computation time of evaluating the monotone function of the system. 
  Furthermore, if the system satisfies the newly defined in-dominating seed property,
  the complexity drops to $O((n+m) \log n \cdot \tau_\omega)$ time.
  We prove that standard $k$-cores in undirected graphs satisfy this property,
  enabling MinRS enumeration in $O((n+m)\log n)$ time,
  a significant improvement over the baseline.
  We also extend our framework to enumerate all solutions in a given monotone system. This yields an $O((n+m)\log n)$-delay algorithm 
  for all $k$-core subgraphs, outperforming an algorithm given by [Boley et al., {\it Theoretical Computer Science}, 2010].
  Our framework is applicable to various $k$-core extensions,
  including weighted $k$-cores, multi-layer $\bm{k}$-cores, and $(k,\ell)$-cores.
\end{abstract}

\section{Introduction}
\label{sec:intro}

We denote by $\bbZ$ the set of integers.
For $i,j\in\bbZ$ $(i\le j)$,
we let $[i,j]\coloneqq\{i,i+1,\dots,j\}$. 
Let $G = (V,E)$ be an undirected graph
with a vertex set $V$ and an edge set $E$.
We denote $n=|V|$ and $m=|E|$. 
The degree of a vertex $v \in V$ is the number of edges in $G$ that are incident to $v$ and is denoted by $\deg_G(v)$.
We denote by $N_G(v)$ the set of open neighbors of $v$.
For a subset $S \subseteq V$,
we denote by $G[S]$ the subgraph of $G$ induced by $S$
and by $G-S$ the subgraph induced by $V \setminus S$,
i.e., $G[V\setminus S]$.
If $S = \{v\}$, then
we write $G - \{v\}$ as $G - v$ for simplicity.
A subset $S\subseteq V$ (or $G[S]$)
is called the \emph{$k$-core} (\emph{of $G$})
if $G[S]$ is the inclusion-wise maximal induced subgraph
such that the minimum degree is
no less than $k$,
where the $k$-core of $G$ is unique~\cite{S.1983}. 
We call $G$ a \emph{$k$-core graph}
if the $k$-core of $G$ is $G$ itself
(i.e., the minimum degree of $G$ is no less than $k$). 
We regard a null graph as a $k$-core graph. 



%
The $k$-core of a graph $G$ can be computed in linear time by a simple iterative procedure~\cite{BZ.2003}
and has various applications in network analyses~\cite{MGPV.2020};
e.g., social networks~\cite{S.1983,SEF.2018,TLL.2020,UKBM.2011}, 
biological networks~\cite{EAAAM.2015,FBST.2024,IS.2015}. 
For $v\in V$, the $k$-core of $G-v$
may be significantly smaller than that of $G$. 
This means that the $k$-core of a real-world network is
not necessarily robust against attacks, disasters, or failures.
When a network is partially destroyed,
the $k$-core of the remaining graph may be much smaller
than the original one, or it may even be empty.
This phenomenon, known as \emph{core collapse},
indicates the vulnerability of $k$-cores~\cite{BDLMG.2015,GDM.2006}.


\figref{example} illustrates this issue.
It shows a graph $G$ and
the 3-cores of some proper subgraphs
obtained by removing a vertex from $G$.
Note that the 3-core of $G$ is $G$ itself. 
As shown in (ii),
the 3-core of $G - v_3$ is $G - v_3$ itself. 
On the other hand, in (iii) and (iv),
the 3-core of $G - v$ for $v \in \{ v_1, v_8 \}$
is not necessarily equal to $G - v$.
In (iii), the 3-core of $G - v_8$ is $G - \{v_8, v_9\}$.
In (iv),
the 3-core of $G - v_1$ consists of only
four vertices $\{v_2, v_3, v_5, v_6\}$.
Furthermore, the 3-core of $G - v_2$ is empty. 

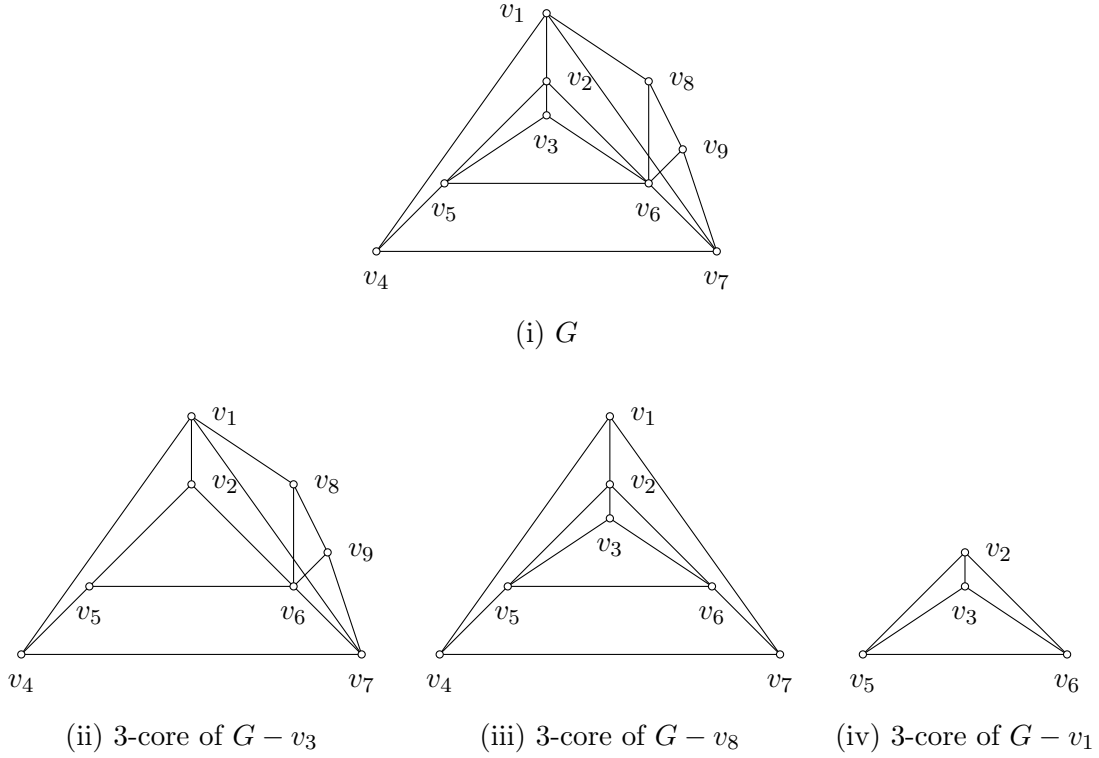
\begin{figure}[t!]
  \centering
  \begin{tabular}{ccc}
    \multicolumn{3}{c}{
  \begin{tikzpicture}[scale=0.45]
    \coordinate (v1) at (5,6) node at (v1) [label=left:$v_1$] {};
    \coordinate (v2) at (5,4) node at (v2) [label=right:$v_2$] {};
    \coordinate (v3) at (5,3) node at (v3) [label=below:$v_3$] {};
    \coordinate (v4) at (0,-1)  node at (v4) [label=below:$v_4$] {};
    \coordinate (v5) at (10,-1) node at (v5) [label=below:$v_7$] {};
    \coordinate (v6) at (2,1) node at (v6) [label=below:$v_5$] {};
    \coordinate (v7) at (8,1) node at (v7) [label=below:$v_6$] {};
    \coordinate (v8) at (8,4) node at (v8) [label=right:$v_8$] {};
    \coordinate (v9) at (9,2) node at (v9) [label=right:$v_9$] {};
    \draw (v1)--(v2);
    \draw (v1)--(v4);
    \draw (v1)--(v5);
    \draw (v2)--(v3);
    \draw (v2)--(v6);
    \draw (v2)--(v7);
    \draw (v3)--(v6);
    \draw (v3)--(v7);
    \draw (v6)--(v7);
    \draw (v4)--(v5);
    \draw (v4)--(v6);
    \draw (v5)--(v7);
    \draw (v8)--(v1);
    \draw (v8)--(v7);
    \draw (v8)--(v9);
    \draw (v9)--(v7);
    \draw (v9)--(v5);
    \fill[fill=white, draw] (v1) circle (3pt) (v2) circle (3pt);
    \fill[fill=white, draw] (v3) circle (3pt) (v4) circle (3pt);
    \fill[fill=white, draw] (v5) circle (3pt) (v6) circle (3pt);
    \fill[fill=white, draw] (v7) circle (3pt) (v8) circle (3pt) (v9) circle (3pt);
  \end{tikzpicture}
  }
    \\
    \multicolumn{3}{c}{(i) $G$}
    \\
    \\
  \begin{tikzpicture}[scale=0.45]
    \coordinate (v1) at (5,6) node at (v1) [label=right:$v_1$] {};
    \coordinate (v2) at (5,4) node at (v2) [label=right:$v_2$] {};
    \coordinate (v4) at (0,-1)  node at (v4) [label=below:$v_4$] {};
    \coordinate (v5) at (10,-1) node at (v5) [label=below:$v_7$] {};
    \coordinate (v6) at (2,1) node at (v6) [label=below:$v_5$] {};
    \coordinate (v7) at (8,1) node at (v7) [label=below:$v_6$] {};
    \coordinate (v8) at (8,4) node at (v8) [label=right:$v_8$] {};
    \coordinate (v9) at (9,2) node at (v9) [label=right:$v_9$] {};
    \draw (v1)--(v2);
    \draw (v1)--(v4);
    \draw (v1)--(v5);
    \draw (v2)--(v6);
    \draw (v2)--(v7);
    \draw (v6)--(v7);
    \draw (v4)--(v5);
    \draw (v4)--(v6);
    \draw (v5)--(v7);
    \draw (v8)--(v1);
    \draw (v8)--(v7);
    \draw (v8)--(v9);
    \draw (v9)--(v7);
    \draw (v9)--(v5);
    \fill[fill=white, draw] (v1) circle (3pt) (v2) circle (3pt);
    \fill[fill=white, draw] (v4) circle (3pt);
    \fill[fill=white, draw] (v5) circle (3pt) (v6) circle (3pt);
    \fill[fill=white, draw] (v7) circle (3pt) (v8) circle (3pt) (v9) circle (3pt);
  \end{tikzpicture}
  &\begin{tikzpicture}[scale=0.45]
    \coordinate (v1) at (5,6) node at (v1) [label=right:$v_1$] {};
    \coordinate (v2) at (5,4) node at (v2) [label=right:$v_2$] {};
    \coordinate (v3) at (5,3) node at (v3) [label=below:$v_3$] {};
    \coordinate (v4) at (0,-1)  node at (v4) [label=below:$v_4$] {};
    \coordinate (v5) at (10,-1) node at (v5) [label=below:$v_7$] {};
    \coordinate (v6) at (2,1) node at (v6) [label=below:$v_5$] {};
    \coordinate (v7) at (8,1) node at (v7) [label=below:$v_6$] {};
    \draw (v1)--(v2);
    \draw (v1)--(v4);
    \draw (v1)--(v5);
    \draw (v2)--(v3);
    \draw (v2)--(v6);
    \draw (v2)--(v7);
    \draw (v3)--(v6);
    \draw (v3)--(v7);
    \draw (v6)--(v7);
    \draw (v4)--(v5);
    \draw (v4)--(v6);
    \draw (v5)--(v7);
    \fill[fill=white, draw] (v1) circle (3pt) (v2) circle (3pt);
    \fill[fill=white, draw] (v3) circle (3pt) (v4) circle (3pt);
    \fill[fill=white, draw] (v5) circle (3pt) (v6) circle (3pt);
    \fill[fill=white, draw] (v7) circle (3pt); 
  \end{tikzpicture}
  &
  \begin{tikzpicture}[scale=0.45]
    \coordinate (v2) at (5,4) node at (v2) [label=right:$v_2$] {};
    \coordinate (v3) at (5,3) node at (v3) [label=below:$v_3$] {};
    \coordinate (v6) at (2,1) node at (v6) [label=below:$v_5$] {};
    \coordinate (v7) at (8,1) node at (v7) [label=below:$v_6$] {};
    \draw (v2)--(v3);
    \draw (v2)--(v6);
    \draw (v2)--(v7);
    \draw (v3)--(v6);
    \draw (v3)--(v7);
    \draw (v6)--(v7);
    \fill[fill=white, draw] (v2) circle (3pt);
    \fill[fill=white, draw] (v3) circle (3pt);
    \fill[fill=white, draw] (v6) circle (3pt);
    \fill[fill=white, draw] (v7) circle (3pt);
  \end{tikzpicture}
\\
  (ii) 3-core of $G-v_3$ & (iii) 3-core of $G-v_8$ & (iv) 3-core of $G-v_1$
  \end{tabular}
  \caption{3-cores of proper subgraphs of $G$}
  \label{fig:example}
\end{figure}

For a $k$-core graph $G=(V,E)$,
let $Y_1,Y_2,\dots,Y_q\subseteq V$ denote
the minimal nonempty subsets of $V$
such that $G-Y_i$, $i\in[1,q]$ is a $k$-core graph.
We claim that computing $Y_1,Y_2,\dots,Y_q$
would provide useful information
on evaluating the robustness of the $k$-core $G$
against vertex removals. 
For example, the size $|Y_i|$ indicates
the degree of core collapse when $v\in Y_i$
is removed from $G$, where it is easy to see
that $Y_i\cap Y_j=\emptyset$ holds for every $i,j\in[1,q]$ $(i\ne j)$.
Furthermore, for $v\in V\setminus(Y_1\cup Y_2\cup\dots\cup Y_q)$,
the $k$-core of $G-v$ is not any
of $G-Y_i$, $i\in[1,q]$;
it is a proper subset of $G-Y_j$ for some $j\in[1,q]$,
indicating that the removal of $v$
yields a larger collapse than any vertex in $Y_j$. 

In \figref{example}, such minimal nonempty subsets of $V$
are $Y_1=\{v_3\}$ and $Y_2=\{v_8,v_9\}$,
where the 3-cores of $G-Y_1$ and $G-Y_2$ are
shown in (ii) and (iii),
respectively.
For $v\in V\setminus(Y_1\cup Y_2)=\{v_1,v_2,v_4,v_5,v_6,v_7\}$, 
the $k$-core of $G-v$ is not any of $G-Y_1$ or $G-Y_2$;
it is a proper subset of $G-Y_1$ and/or $G-Y_2$;
e.g., the 3-core of $G-v_1$ in (iv)
is a proper subgraph of $G-Y_2$. 

Let us introduce terminologies to formalize a problem based on the above motivation. 
Suppose that an element set $U$ and a property $\Pi$
are given. Let $S\subseteq U$. 
A nonempty subset $Y\subseteq S$
is a \emph{removable set of $S$}
if $S\setminus Y$ satisfies $\Pi$,
where whether or not $S$ satisfies $\Pi$ does not matter. 
We abbreviate a removable set into an RS for simplicity
and call an inclusion-wise minimal RS a \emph{MinRS}.
For example, when we discuss $k$-cores, 
a MinRS of a graph $G=(V,E)$
is a minimal nonempty subset $Y\subseteq V$
such that $G-Y$ is a $k$-core graph. 

We formalize the problem \probMRS: 
Given a $k$-core graph $G=(V,E)$, find all MinRSs. 
Obviously,
the problem can be solved by the following simple algorithm:
compute the MinRS of $G-v$
for each $v\in V$ and collect the minimal subsets
among them. This algorithm takes at least $\Omega((n+m)n+\Tcol)$ time,
where $\Tcol$ denotes the time for collecting minimal subsets
and $\Tcol=O(n^3)$ holds. 


In this paper, we show that 
the problem \probMRS\ %
can be solved in 
$O((n+m)\log n)$ time
under a general framework based on 
what we call a monotone system.  
For a ground set $U$ and $\MF\subseteq 2^U$,
we call $(U,\MF)$ a \emph{set system},
or a system for short.
When $U$ is clear from the context, we may regard
$\MF$ as a system for simplicity. 
Let $\MZ=(Z,\leqslant)$ be a poset; 
$\omega: 2^U\times U\to Z$ be a function;
and $\theta\in Z$. 
For $S\subseteq U$ and $u\in U$, 
we may state $\omega(S,u)$ as $\omega_S(u)$. 
Given a tuple $(U,\MZ,\omega,\theta)$, 
we call $S\subseteq U$ a \emph{$\theta$-solution} if
$\theta\leqslant\omega_S(u)$ holds for every $u\in S$.
We regard an empty set as a $\theta$-solution.
We define 
$\MS(U,\MZ,\omega,\theta)\coloneqq\{S\subseteq U\mid S\text{ is a }\theta\text{-solution}\}$
to be the family of all $\theta$-solutions.
\begin{itemize}
  \item 
We say that $\omega$ is \emph{monotone} if
$\omega_S(u)\leqslant\omega_{S'}(u)$ holds
for any $S\subseteq S'\subseteq U$ and $u\in U$.
We call $(U,\MS(U,\MZ,\omega,\theta))$
(or $\MS(U,\MZ,\omega,\theta)$)
a \emph{monotone system} if $\omega$ is monotone.
\item 
Assuming an underlying graph $G=(V,E)$, 
we say that $\omega$ is \emph{local}
(\emph{with respect to $G$})~\cite{BZ.2002,BatZav.2011}
if $\omega_S(u)=\omega_{S\cap N_G(u)}(u)$ holds
for any $S\subseteq V$ and $u\in V$,
that is, $\omega_S(u)$ depends only on the neighbors of $u$
that belong to $S$.
We call $(V,\MS(V,\MZ,\omega,\theta))$ (or $\MS(V,\MZ,\omega,\theta)$)
\emph{local monotone} if $\omega$ is local and monotone.
To emphasize $G$,
we may state it as $\MS(V,\MZ,\omega,\theta;G)$. 
\end{itemize}



Suppose that we are given a
local monotone system $\MS=\MS(V,\MZ,\omega,\theta;G)$.
We consider the problem of finding all MinRSs $Y$ of $V$
such that $V\setminus Y$ is a $\theta$-solution. 
As we will see, the problem is a generalization of \probMRS. 
Let us denote by $\tau_\omega$ and $\kappa_\omega$
upper bounds on the time and space complexity
to compute $\omega_S(u)$ for $S\subseteq V$ and $u\in V$,
respectively. Our first theorem is summarized as follows. 

\begin{thm}
  \label{thm:main_local_system}
  For a local monotone system $\MS(V,\MZ,\omega,\theta; G)$,
  we can enumerate all MinRSs of $V$ in
  $O((n+m)n\cdot\tau_\omega)$ time
  and $O(n+m+\kappa_\omega)$ space.
\end{thm}
Note that whether $V$ itself is a $\theta$-solution
or not does not matter. 
To show the theorem, we examine mathematical
properties of a monotone system
$\MS=\MS(V,\MZ,\omega,\theta)$, not limited to
local monotone systems, and design an algorithm
to enumerate all MinRSs of $V$, where
we exploit two auxiliary digraphs, that is
propagation digraphs and seed-relation (SR) digraphs.
Then we analyze the computational complexity of the algorithm
when $\MS$ is local monotone with respect to
an underlying graph $G$, yielding \thmref{main_local_system}.

Next, we show that the time complexity can be
improved if a local monotone system satisfies
what we call the in-dominating seed property. 
\begin{thm}
  \label{thm:in-dominating}
  For a local monotone system $\MS=\MS(V,\MZ,\omega,\theta; G)$,
  if $\MS$ satisfies the in-dominating seed property,
  then we can enumerate all MinRSs of $V$ in
  $O((n+m)\log n\cdot\tau_\omega)$ time
  and $O(n+m+\kappa_\omega)$ space.
\end{thm}

Our framework can be applied to a lot of extensions of $k$-cores~\cite{MGPV.2020};
e.g., generalized cores~\cite{BZ.2002,BatZav.2011},
$(k,\ell)$-cores in digraphs~\cite{GTV.2013},
weighted $k$-cores~\cite{GSH.2012}, 
multi-layer $\bm{k}$-cores~\cite{GBG.2017}.
For these extensions, we show that all MinRSs of $V$
can be generated in $O((n+m)n)$ time
by \thmref{main_local_system},
whereas the baseline algorithm requires at least
$\Omega((n+m)n+\Tcol)$ time. 
Furthermore, for $k$-cores in undirected graphs,
we can show the following theorem, using
the in-dominating seed property and \thmref{in-dominating}. 

\begin{thm}
  \label{thm:deg-at-least-k}
  The problem \probMRS\ can be solved  
  in $O((n+m)\log n)$ time
  and $O(n+m)$ space.
\end{thm}

For algorithmic studies on core collapse,
given an undirected graph $G$ and integers $b,x,k$,
the problem {\sc Collapsed $k$-Core} asks for whether
there is a subset $S$ of at most $b$ vertices
such that the $k$-core of $G-S$ has at most $x$ vertices. 
The problem is NP-hard~\cite{ZZQZL.2017}, 
and fixed-parameter tractability~\cite{LMS.2021} and
mathematical optimization models~\cite{CSSAL.2023} are investigated. 
The papers~\cite{HSLS.2023,LZWX.2024,LZWYX.2024,ZLWZX.2023} study problems on core collapse 
caused by edge removal. 
For example, Honcharov et al.~\cite{HSLS.2023} study 
the problem of finding a spanning subgraph with 
the minimum number of edges
such that
the core number of each vertex is the same as in the input graph.


The paper is organized as follows.
We prepare notations and terminologies in \secref{prel}.
In \secref{mmw}, 
we discuss monotone systems to
analyze their mathematical properties 
and to present an algorithm that enumerates
all MinRSs of the ground set of a given system.
Then in \secref{lms},
we explain how we achieve computational efficiency
for local monotone systems
and those satisfying
the in-dominating seed property,
along with the proofs for Theorems~\ref{thm:main_local_system}
and~\ref{thm:in-dominating}. 
We also describe how we apply the framework
to $k$-cores and the extensions, where we provide a proof for \thmref{deg-at-least-k}.
Finally, we consider the problem of 
enumerating ``all'' $\theta$-solutions in a monotone system.
For the enumeration of all $k$-core subgraphs in a given graph, our algorithm achieves a better time bound than an algorithm based on Boley et al.'s framework~\cite{BHPW.2010}. 
We conclude the paper in \secref{conc}.

\section{Preliminaries}
\label{sec:prel}

Throughout the paper,
we use $V$ to represent the vertex set of a graph
or the ground set of elements.
We use the terms ``maximal'' or ``minimal''
to indicate ``inclusionwise-maximal'' or ``inclusionwise-minimal'' if no confusion arises.
If the domain of a function $f$ is a family of subsets
(e.g., $2^V$ for a ground set $V$),  
we write $f(\{x_1,x_2,\dots,x_k\})$
as $f(x_1,x_2,\dots,x_k)$ for simplicity.

For a ground set $V$, 
we represent a subset of $V$ by italic
(e.g., $A\subseteq V$);
a family of subsets of $V$ by calligraphic
(e.g., $\MA\subseteq 2^V$); and
a collection of families of subsets of $V$ by
blackboard bold (e.g., $\bbA\subseteq 2^{2^V}$).
We use the union operator $\cup$
to represent descending the set hierarchy
(i.e., $\cup\bbA=\bigcup_{\MA\in\bbA}\MA\subseteq 2^V$,
$\cup\MA=\bigcup_{A\in\MA}A\subseteq V$). 



For a graph $G$, let $S$ be a subset of vertices. 
We may abbreviate $G[S]$ into $S$ for simplicity
when $G[S]$ appears in a subscript
if no confusion arises;
e.g., $N_{G[S]}(v)$ may be written as $N_S(v)$
for a vertex $v$.  

For an undirected graph $G=(V,E)$, let $v\in V$.
we denote by $N_G[v]=N_G(v)\cup\{v\}$
the set of closed neighbors. 
For $S\subseteq V$,
we let $N_G(S)\coloneqq\big(\bigcup_{v\in S}N_G(v)\big)\setminus S$
and $N_G[S]\coloneqq N_G(S)\cup S$.

For a directed graph (or a digraph) $D=(V,A)$,
let $v\in V$. 
We denote the set of open out-neighbors (resp., in-neighbors)
of $v$ by $N^+_D(v)$ (resp., $N^-_D(v)$).
We also denote the set of closed out-neighbors (resp., in-neighbors) of $v$
by $N^+_D[v]\coloneqq N^+_D(v)\cup\{v\}$
(resp., $N^-_D[v]\coloneqq N^-_D(v)\cup\{v\}$).
For convenience, we define $N_D(v)\coloneqq N^+_D(v)\cup N^-_D(v)$ and $N_D[v]\coloneqq N^+_D[v]\cup N^-_D[v]$.
We denote by $R_D(v)$ 
the set of vertices that are reachable from $v$.  
We regard $v\notin R_D(v)$ and
let $R_D[v]\coloneqq R_D(v)\cup\{v\}$. 
For $S\subseteq V$, we let $R_D(S)\coloneqq\big(\bigcup_{v\in S}R_D(v)\big)\setminus S$ and $R_D[S]\coloneqq R_D(S)\cup S$. 
We abbreviate a weakly connected component and a strongly connected component
into WCC and SCC, respectively. 
In the SCC decomposition, 
we call an SCC that has no outgoing arcs
a \emph{bottom SCC},
and a vertex in a bottom SCC a \emph{bottom vertex}.



Let $\MZ=(Z,\leqslant)$ be a partially ordered set (poset)
that consists of a ground set $Z$ and a partial order $\leqslant$ on $V$.
For $x,y\in Z$, we may write $y\geqslant x$ if $x\leqslant y$.
A mapping $\cl:Z\to Z$ is a \emph{closure operator}
(\emph{on $\MP$})
if it satisfies the extensivity (i.e., $x\leqslant\cl(x)$),
the monotonicity (i.e., $x\leqslant y$ $\Rightarrow$ $\cl(x)\leqslant\cl(y)$), and
the idempotence (i.e., $\cl(\cl(x))=\cl(x)$).
The element $\cl(x)$ is called the \emph{closure of $x$}. 


\section{Monotone Systems}
\label{sec:mmw}

In this section, we examine mathematical properties
of monotone systems $\MS(V,\MZ,\omega,\theta)$ in \secref{mmw_prop}
and design an algorithm to enumerate all MinRSs of $V$ in \secref{mono_alg}. 


\subsection{Mathematical Properties}
\label{sec:mmw_prop}
For a monotone system $\MS(V,\MZ,\omega,\theta)$,
let $S\subseteq V$ be a non-empty subset. 
If $S$ is not a $\theta$-solution, then there is $u\in S$
such that $\omega_S(u) \not\geqslant \theta$. 
We say that $u$ is \emph{violating $S$}. 
Also there is an RS $Y\subseteq S$ of $S$
such that $S\setminus Y$ is a $\theta$-solution;
recall that we regard $\emptyset$ as a $\theta$-solution. 
A violating element $u$ belongs to $Y$
since otherwise we would have
$\omega_{S}(u)\geqslant\omega_{S\setminus Y}(u)\geqslant\theta$, contradicting
$\omega_{S}(u)\not\geqslant\theta$.
In particular, if $Y$ is a MinRS of $S$, then
$S\setminus Y'$ is not a $\theta$-solution for any
$Y'\subsetneq Y$ and 
there is $v\in Y\setminus Y'$ such that $\omega_{S\setminus Y'}(v) \not\geqslant \theta$. 
This simple observation is used in the following discussions. 

The following lemma generalizes
the fact that, if an undirected graph $G$
is not a $k$-core graph, then
the $k$-core of $G$ is unique, or
equivalently, the MinRS of $G$ is unique,
to the context of monotone systems. 

\begin{lem}
  \label{lem:unique}
  Given a monotone system
  $\MS(V,\MZ,\omega,\theta)$, 
  let $S\subseteq V$.
  The maximal $\theta$-solution
  among subsets of $S$ is unique.   
\end{lem}
\begin{proof}
  If $S$ itself is a $\theta$-solution, then we are done.
  Otherwise, suppose that there are two different MinRSs of $S$, 
  say $Y_1,Y_2\subseteq S$, for contradiction.
  There is $u\in S$ such that
  $\omega_S(u)\not\geqslant\theta$ since $S$ is not a $\theta$-solution.
  Both $Y_1$ and $Y_2$ contain such $u$ 
  and hence $u\in Y_1\cap Y_2$ holds.
  By the minimality,
  neither $Y_1\setminus Y_2$ nor $Y_2\setminus Y_1$ is empty.
  The subset $S\setminus (Y_1\cap Y_2)$
  is not a $\theta$-solution,
  and there is an element $v\in Y_1\setminus Y_2$ such that
  $\omega_{S\setminus(Y_1\cap Y_2)}(v) \not\geqslant \theta$.
  However, this would yield
  $\omega_{S\setminus(Y_1\cap Y_2)}(v)\geqslant\omega_{S\setminus Y_2}(v)\geqslant\theta$, 
  a contradiction.  
\end{proof}

For $C\subseteq V$, let $Y$ denote the unique superset of $C$
such that $V\setminus Y$ is the maximal $\theta$-solution among subsets of $V\setminus C$.
The uniqueness of such $Y$ is guaranteed by \lemref{unique}.
We define a function $\rho:2^V\to2^V$ to be $\rho(C)\triangleq Y$.
In other words, when $C\ne\emptyset$, $\rho(C)$ is the unique RS of $V$
that is minimal among supersets of $C$.
Note that $\rho(C)$ is an RS of $V$, but not necessarily a MinRS of $V$.

\begin{lem}
  \label{lem:closure}
  Given a monotone system
  $\MS(V,\MZ,\omega,\theta)$, the function $\rho$ is a closure operator on the poset
  $(2^V,\subseteq)$. 
\end{lem}
\begin{proof}
  It is immediate from the definition that $\rho$ satisfies
  the extensivity (i.e., $C\subseteq\rho(C)$) and
  the idempotence (i.e., $\rho(\rho(C))=\rho(C)$).
  For the monotonicity,
  let $C_1,C_2\subseteq V$. 
  We show that $\rho(C_2)\subseteq\rho(C_1)$ holds if $C_2\subseteq C_1$. 
  Note that $C_2\subseteq\rho(C_1)$ holds by the extensivity.

  If $\rho(C_1)=V$, then 
  we have $\rho(C_1)=V\supseteq\rho(C_2)$. 
  Suppose $\rho(C_1)\subsetneq V$ and $C_2=\emptyset$.
  If $V$ is a $\theta$-solution, then $\rho(C_2)=\emptyset\subseteq\rho(C_1)$ holds. Otherwise, $\rho(C_2)=\rho(\emptyset)$ is the unique MinRS of $V$.
  Any RS of $V$ should contain $\rho(C_2)$ as a subset, and hence $\rho(C_2)\subseteq\rho(C_1)$. 

  Let us consider the remaining case;
  suppose $\rho(C_1)\subsetneq V$ and $C_2\ne\emptyset$.
  By $C_2\subseteq\rho(C_1)$ and $C_2\subseteq\rho(C_2)$,
  we have $C_2\subseteq\rho(C_1)\cap\rho(C_2)$, implying that
  $\rho(C_1)\cap\rho(C_2)$ is not empty. 
  Suppose $\rho(C_2)\not\subseteq\rho(C_1)$ 
  for contradiction.
  For any $C_2\subseteq R\subsetneq\rho(C_2)$, $V\setminus R$ is not a $\theta$-solution. 
  Then $V\setminus(\rho(C_1)\cap\rho(C_2))$ is not a $\theta$-solution
  since $C_2\subseteq \rho(C_1)\cap\rho(C_2)\subsetneq\rho(C_2)$,
  where the last inclusion relationship holds by 
  $\rho(C_2)\setminus\rho(C_1)\ne\emptyset$. 
  There is an element $v\in\rho(C_2)\setminus\rho(C_1)$
  such that $\omega_{V\setminus(\rho(C_1)\cap\rho(C_2))}(v)\not \geqslant \theta$.
  However, we would have $\omega_{V\setminus(\rho(C_1)\cap\rho(C_2))}(v)\geqslant\omega_{V\setminus\rho(C_1)}(v)\geqslant \theta$, a contradiction. 
\end{proof}

We can compute the closure $\rho(C)$ of $C$
by a simple iterative procedure in \algref{rho}. 
This is regarded as a generalization of the algorithm
that computes the $k$-core of a given graph.

\begin{algorithm}[t!]
  \caption{An algorithm to compute the closure $\rho(C)$ of a subset $C\subseteq V$}
  \label{alg:rho}
  \Input{A subset $C\subseteq V$ in a monotone system $\MS(V,\MZ,\omega,\theta)$}
  \Output{The closure $\rho(C)$ of $C$}
  $Y\gets C$\;
  \lWhile{there is $u\in V\setminus Y$ such that $\omega_{V\setminus Y}(u)\not\geqslant \theta$}
    {
      $Y\gets Y\cup\{u\}$
    }
  output $Y$ as $\rho(C)$\;
\end{algorithm}

\paragraph{Propagation Digraphs.}
To understand the structure of RSs of $V$ in
a monotone system $\MS\coloneqq\MS(V,\MZ,\omega,\theta)$,
we introduce the \emph{propagation digraph} $H^{\propag}_\MS=(V,A^{\propag}_\MS)$ that consists of the vertex set $V$
and the arc set 
$A^{\propag}_\MS\triangleq\{(u,v)\in V\times V\mid \omega_{V\setminus\{u\}}(v)\not\geqslant \theta\}$.
The digraph shows the propagation of
weights falling below the threshold $\theta$
when one element is removed from $V$. 
Any path from a vertex $u$ to other vertex 
represents that, if $u$ is removed from $V$,
then all the other vertices on the path should be removed
to obtain a $\theta$-solution,
that is,
any RS of $V$ that contains $u$ must contain
all vertices on the path,
whereas they may not be sufficient.  
Thus, for the purpose of enumerating MinRSs of $V$,
bottom SCCs in the digraph must be important
since, for every vertex $u$,
there is a bottom SCC which is reachable from $u$.
Let $u\in V$ and $C\subseteq V$.
For notational convenience, 
let us write $R_{H^{\propag}_{\MS}}[u]$
(resp., $R_{H^{\propag}_{\MS}}[C]$; 
$N^+_{H^{\propag}_\MS}(u)$; and $N^+_{H^{\propag}_\MS}[u]$)
as $R^{\propag}_{\MS}[u]$
(resp., $R^{\propag}_{\MS}[C]$; $N^{\propag+}_\MS(u)$; and $N^{\propag+}_\MS[u]$).


We denote by $\MB^{\propag}_\MS\subseteq 2^V$ the family of bottom SCCs
in the propagation digraph $H^{\propag}_\MS$. 
For $B\in\MB^{\propag}_\MS$, if $|B|=1$, say $B=\{u\}$, then
$\{u\}$ is a MinRS of $V$
since its out-degree is zero and hence
$\omega_{V\setminus\{u\}}(v)\geqslant\theta$ holds for all $v\in V\setminus\{u\}$,
indicating that $V\setminus\{u\}$ is a $\theta$-solution.
If $|B|\ge2$, then $B$ is not necessarily an RS;
for $B=\{u_1,u_2,\dots,u_k\}$, 
it is possible that $\omega_{V\setminus B}(v)\not\geqslant \theta$ holds for $v\in V\setminus B$
even when $\omega_{V\setminus\{u_i\}}(v)\geqslant\theta$, $i\in[1,k]$.
If $V$ is not a $\theta$-solution,
then the bottom SCC is unique since
each violating vertex has an incoming edge
from any other vertex. 


Using bottom SCCs, 
we have the following characterization of MinRSs. 

\begin{lem}
  \label{lem:minRS_necsuf}
  Suppose that we are given
  a monotone system $\MS=\MS(V,\MZ,\omega,\theta)$.
  A subset $Y\subseteq V$
  is a MinRS of $V$
  if and only if
  {\rm(i)} there is a bottom SCC $B\in\MB^{\propag}_\MS$
  such that $Y=\rho(B)$; and
  {\rm(ii)} for any bottom SCC $B'\in\MB^{\propag}_\MS$,
  $Y\ne\rho(B')$ implies $Y\not\supseteq\rho(B')$. 
\end{lem}
\begin{proof}
  For the necessity,
  $Y$ must contain all bottom SCCs $B\in\MB^{\propag}_\MS$
  that are reachable from $Y$,
  and such $B$ always exists. 
  Then $Y=\rho(Y)\supseteq\rho(B)$ holds.
  The closure $\rho(B)$ is an RS of $V$
  and, by the minimality of $Y$,
  $Y=\rho(B)$ holds, showing (i).
  (ii) is immediate.

  For the sufficiency, we have $Y=\rho(B)$ by (i)
  and $\rho(B)$ is an RS of $V$ by the definition of $\rho$.
  If $\rho(B)$ is not a MinRS of $V$,
  then there is a MinRS $Y'\subsetneq\rho(B)$ that is a proper subset of $\rho(B)$.
  By the necessity of MinRS,
  there would be $B'\in\MB^{\propag}_\MS$ such that
  $\rho(B')=Y'$,
  yielding $\rho(B')=Y'\subsetneq\rho(B)=Y$,
  contradicting (ii). 
\end{proof}

\invis{  
  \begin{itemize}
  \item[\rm (i)]
    For any MinRS $Y$ of $V$,
    there is a bottom SCC $B\in\MB^{\propag}_\MS$
    of $H^{\propag}_{\MS}$ such that $Y=\rho(B)$. 
  \item[\rm (ii)] Let $B\in\MB^{\propag}_\MS$ be a bottom SCC
    of $H^{\propag}_{\MS}$.
    The closure $\rho(B)$ is a MinRS of $V$
    if and only if
    no other bottom SCC $B'\in\MB^{\propag}_\MS$ $(B'\ne B)$
    satisfies $\rho(B')\subsetneq\rho(B)$. 
  \end{itemize}
\begin{proof}
  (i) The MinRS $Y$ must contain
  all bottom SCCs in $H^{\propag}_\MS$
  that are reachable from a vertex in $Y$. That is,
  $Y\supseteq B$ holds for all $B\in\MB^{\propag}_\MS$.
  Then $Y=\rho(Y)\supseteq\rho(B)$ holds.
  The closure $\rho(B)$ is an RS of $V$
  and, by the minimality of $Y$,
  $Y=\rho(B)$ holds.   

  \noindent (ii)
  By the definition, $\rho(B)$ is an RS.
  For the contraposition,
  if $\rho(B)$ is not a MinRS,
  then there is a MinRS $Y\subsetneq\rho(B)$ that is a proper subset of $\rho(B)$.
  By (i), there is $B''\in\MB^{\propag}_\MS$ such that $B''\subseteq Y$,
  yielding $\rho(B'')\subseteq\rho(Y)=Y\subsetneq\rho(B)$.  
\end{proof}
}

We see that any MinRS of $V$ can be generated
by computing the closure $\rho(B)$ of some bottom SCC
$B\in\MB^{\propag}_\MS$. 
Thus we call a bottom SCC in the propagation graph
a \emph{seed} and vertices in a seed \emph{seed vertices}.
We can enumerate all MinRSs of $V$ as follows:
(1) Construct the propagation digraph $H^{\propag}_\MS$.
(2) Compute the seeds, that is, the bottom SCCs of $H^{\propag}_\MS$.
(3) For each seed $B\in\MB^{\propag}_\MS$, compute $\rho(B)$ by \algref{rho}.
(4) Collect the minimal subsets in
$\{\rho(B)\mid B\in\MB^{\propag}_\MS\}$;
they are MinRSs of $V$ by \lemref{minRS_necsuf}.
However, by this approach,
it must be hard to achieve a better time bound
than the baseline algorithm mentioned in \secref{intro}
since the number of seeds
is $|\MB^{\propag}_\MS|=O(n)$. 


\subsection{Algorithms to Enumerate MinRSs}
\label{sec:mono_alg}
Again, let $\MS\coloneqq\MS(V,\MZ,\omega,\theta)$
be a monotone system. 
In this subsection, we describe an alternative algorithm
to enumerate all MinRSs of $V$.
As shown in the next section, 
this algorithm achieves a better time bound than
the baseline algorithm when $\MS$ is local monotone.


In the proposed algorithm,
we use a partition of the seed set $\MB^{\propag}_\MS$
to maintain candidates of seeds
whose closures are MinRSs of $V$. 
For a subset $\MB\subseteq\MB^{\propag}_\MS$ of seeds,
we define $\Sigma(\MB)\coloneqq\{B^\ast\in\MB\mid\forall B\in\MB,\ \rho(B^\ast)\subseteq\rho(B)\}$. 
We call $\MB$ \emph{pointed} if $\Sigma(\MB)\ne\emptyset$.
When $\MB$ is pointed,
$\rho(B^\ast)=\rho(B^{\ast\ast})$ holds for any $B^\ast,B^{\ast\ast}\in\Sigma(\MB)$,
and we write $\rho(B^\ast)$ or $\rho(B^{\ast\ast})$
as $\rho(\MB)$ for simplicity. 
Let $\bbB$ be a partition of the seed set $\MB^{\propag}_\MS$.
We call $\bbB$ \emph{pointed} if every $\MB\in\bbB$ is pointed. 
\invis{
  for every $i\in[1,q]$,
  there is a seed $B^\ast_i\in\MB_i$ that satisfies
  $\rho(B^\ast_i)\subseteq\rho(B)$ for all $B\in\MB_i$. 
  For a pointed partition $\bbB=\{\MB_1,\MB_2,\dots,\MB_q\}$,
  we define $\Sigma_\bbB(\MB_i)\coloneqq\{B^\ast_i\in\MB_i\mid\forall B\in\MB_i,\ \rho(B^\ast_i)\subseteq\rho(B)\}$, $i\in[1,q]$. 
  It is clear that $\Sigma_\bbB(\MB_i)$ is not empty and
  $\rho(B^\ast_i)=\rho(B^{\ast\ast}_i)$ holds for
  $B^\ast_i,B^{\ast\ast}_i\in\Sigma_\bbB(\MB_i)$.
}
For example, the partition $\bbB_\bot\coloneqq\bigcup_{B\in\MB^{\propag}_\MS}\{\{B\}\}$ 
that consists of singletons of seeds is pointed. 
The following lemma motivates us
to utilize a pointed partition of
the seed set $\MB^{\propag}_\MS$. 

\begin{lem}
  \label{lem:minrs_pointed}
  Suppose that we are given
  a monotone system $\MS=\MS(V,\MZ,\omega,\theta)$.
  Let $\bbB$ 
  be a pointed partition of $\MB^{\propag}_\MS$.
  For any MinRS $Y$ of $V$,
  there exists $\MB\in\bbB$ such that $Y=\rho(\MB)$. 
\end{lem}
\begin{proof}
  By \lemref{minRS_necsuf},
  there is a seed $B\in\MB^{\propag}_\MS$ such that
  $\rho(B)=Y$.
  Let $\MB\in\bbB$ denote the subset in $\bbB$
  such that $B\in\MB$. The partition $\bbB$ is pointed,
  and hence $\MB$ is pointed.
  Any $B^\ast\in\Sigma(\MB)$ satisfies
  $\rho(\MB)=\rho(B^\ast)\subseteq\rho(B)=Y$,
  where the inclusion relationship
  must hold by equality due to the minimality of $Y$. 
\end{proof}

Informally, the proposed algorithm enumerates
MinRSs of $V$ as follows;
starting from the trivial pointed partition $\bbB=\bbB_\bot$,
we repeat updating $\bbB$ by
$\bbB\gets(\bbB\setminus\{\MB,\MB'\})\cup\{\MB\cup\MB'\}$
for $\MB,\MB'\in\bbB$
such that $\rho(\MB)\subseteq\rho(\MB')$
(resp., $\rho(\MB')\subseteq\rho(\MB)$) holds,
where $\MB\cup\MB'$ is pointed in such a case
since $\Sigma(\MB\cup\MB')\supseteq\Sigma(\MB)\ne\emptyset$
(resp., $\Sigma(\MB\cup\MB')\supseteq\Sigma(\MB')\ne\emptyset$).
Finally we obtain a pointed partition $\bbB$
such that, for any $\MB\in\bbB$,
no $\MB'\in\bbB$ $(\MB'\ne\MB)$ satisfies
$\rho(\MB)\supseteq\rho(\MB')$. 
We can show that there is a one-to-one correspondence
between $\bbB$ and the MinRSs of $V$
(i.e., for each MinRS $Y$,
there is a unique $\MB\in\bbB$ such that $\rho(\MB)=Y$),
by which we are done. 

\paragraph{SR Digraphs and 1-SR Digraphs.}
Maintaining
a pointed partition $\bbB=\{\MB_i\}_{i=1}^q$,
the algorithm runs with another auxiliary digraph. 
Let $F^\ast[\bbB]\coloneqq\{(\MB_i,\MB_j)\in\bbB\times\bbB\mid \rho(\MB_i)\supseteq\rho(\MB_j),\ i\ne j\}$.
We call the digraph $H^\ast_\bbB=(\bbB,F^\ast[\bbB])$
the \emph{seed-relation} (\emph{SR}) \emph{digraph of $\bbB$}.
We call a subset $\MB\in\bbB$ of seeds 
a \emph{seed-subset node} in the context of an SR digraph.  
Let $H=(\bbB,F)$ be a spanning subgraph of $H^\ast_\bbB$.
We call $H$ a \emph{one-degree SR} (\emph{1-SR}) \emph{digraph of $\bbB$}
if $\deg^+_H(\MB)\le 1$ holds for every $\MB\in\bbB$; and 
$\deg^+_H(\MB)=1$ holds whenever $\deg^+_{H^\ast_\bbB}(\MB)\ge1$.
In other words, $H$ is obtained by picking up precisely
one outgoing arc for each seed-subset node in $H^\ast_\bbB$
if one exists. 
Observe that every WCC of a 1-SR digraph is
either an in-tree (i.e., an arborescence such that all arcs are reversed)
or a functional graph (i.e., all out-degrees are one). 
For an in-tree WCC, we call the unique node with out-degree zero
the \emph{root} of the WCC.
An example of a 1-SR digraph is shown in \figref{1-SR}. 

\begin{figure}[t!]
  \centering
  \resizebox{\textwidth}{!}{%
    \begin{tikzpicture}[
  >=Latex,
  node distance=1.4cm and 2.0cm,
  box/.style={draw, minimum width=8mm, minimum height=8mm, inner sep=0pt},
  lbl/.style={font=\large},
  wlbl/.style={font=\large},
  flow/.style={->, shorten <=2pt, shorten >=2pt}
]

\node[box] (b11) at (0,0) {};
\node[box] (b12) at (3,0) {};
\node[lbl, above=2mm of b11] {$\{B_{1,1}\}$};
\node[lbl, above=2mm of b12] {$\{B_{1,2}\}$};
\draw[flow] (b11.north east) to[out=30,in=150] (b12.north west);
\draw[flow] (b12.south west) to[out=210,in=330] (b11.south east);
\node[wlbl, text=blue!70!black] at (1.5,-1.0) {$\mathbb{B}_1=\{\{B_{1,1}\},\{B_{1,2}\}\}$};

\node[box] (b13) at (5.0,0) {};
\node[box] (b14) at (7.5,0) {};
\node[lbl, above=2mm of b13] {$\{B_{1,3}\}$};
\node[lbl, above=2mm of b14] {$\{B_{1,4}\}$};
\draw[flow] (b13.north east) to[out=30,in=150] (b14.north west);
\draw[flow] (b14.south west) to[out=210,in=330] (b13.south east);
\node[wlbl, text=blue!70!black] at (6.25,-1.0) {$\mathbb{B}_2=\{\{B_{1,3}\},\{B_{1,4}\}\}$};

\node[box] (b31) at (9.5,0.3) {};
\node[box] (b32) at (11.6,0.3) {};
\node[lbl, above=2mm of b31] {$\{B_{3,1}\}$};
\node[lbl, above=2mm of b32] {$\{B_{3,2}\}$};
\draw[flow] (b31.north east) to[out=25,in=155] (b32.north west);
\draw[flow] (b32.south west) to[out=205,in=335] (b31.south east);
\node[wlbl, text=blue!70!black] at (10.55,-0.85) {$\mathbb{B}_4=\{\{B_{3,1}\},\{B_{3,2}\}\}$};

\node[box] (b33) at (9.5,-3.2) {};
\node[box] (b34) at (11.6,-3.2) {};
\node[lbl, above=2mm of b33] {$\{B_{3,3}\}$};
\node[lbl, above=2mm of b34] {$\{B_{3,4}\}$};
\draw[flow] (b33.north east) to[out=25,in=155] (b34.north west);
\draw[flow] (b34.south west) to[out=205,in=335] (b33.south east);
\node[wlbl, text=blue!70!black] at (10.55,-4.4) {$\mathbb{B}_5=\{\{B_{3,3}\},\{B_{3,4}\}\}$};

\node[box] (b21) at (3.3,-2.8) {};
\node[box] (b22) at (4.8,-4.7) {};
\node[box] (b23) at (2.2,-4.7) {};
\node[lbl, above=1mm of b21] {$\{B_{2,1}\}$};
\node[lbl, right=1mm of b22] {$\{B_{2,2}\}$};
\node[lbl, left=2mm of b23] {$\{B_{2,3}\}$};
\draw[flow] (b21) -- (b22);
\draw[flow] (b23.north east) to[out=25,in=155] (b22.north west);
\draw[flow] (b22.south west) to[out=205,in=335] (b23.south east);
\node[wlbl, text=blue!70!black] at (3.5,-5.8) {$\mathbb{B}_3=\{\{B_{2,1}\},\{B_{2,2}\},\{B_{2,3}\}\}$};

\node[box] (b41) at (14.0,0.1) {};
\node[box] (b51) at (14.0,-2.3) {};
\node[box] (b61) at (14.0,-4.6) {};
\node[lbl, right=1mm of b41] {$\{B_{4,1}\}$};
\node[lbl, right=1mm of b51] {$\{B_{4,2}\}$};
\node[lbl, right=1mm of b61] {$\{B_{5,1}\}$};
\draw[flow] (b41) -- (b51);
\draw[flow] (b51) -- (b61);
\node[wlbl, text=blue!70!black] at (14.0,-5.5) {$\mathbb{B}_6=\{\{B_{4,1}\},\{B_{4,2}\},\{B_{5,1}\}\}$};

\end{tikzpicture}
  }
  \caption{A 1-SR digraph $H=(\bbB,F)$ for $\bbB=\bbB_\bot$,
    where $\MB^{\propag}_\MS=\{B_{1,1},B_{1,2},\dots,B_{5,1}\}$
    is a set of 14 seeds. Each seed constitutes a singleton in $\bbB_\bot$
    and corresponds to a seed-subset node indicated by a square.     
    For $i\in[1,5]$, we denote by $\rho_i$ the closure of a seed $B_{i,j}$;
    e.g., $\rho_1=\rho(B_{1,1})=\dots=\rho(B_{1,4})$.
    Among $\rho_1,\dots,\rho_5$,
    we assume $\rho_2\subsetneq\rho_1$ and $\rho_5\subsetneq\rho_4$ for the inclusion relationship,
    and hence $\rho_2$, $\rho_3$ and $\rho_5$ are MinRSs by \lemref{minRS_necsuf}. 
  }
  \label{fig:1-SR}
\end{figure}
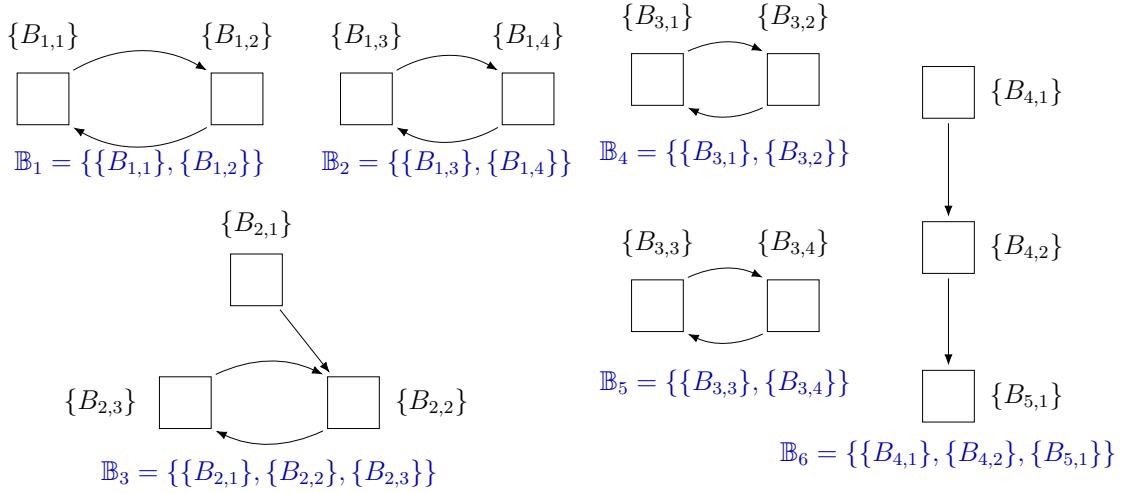

In the proposed algorithm,
we do not deal with SR digraphs explicitly
since its size is $O(n^2)$,
but construct 1-SR digraphs $H=(\bbB,F)$ repeatedly,
updating the pointed partition $\bbB$,
where the size of $H$ is $O(n)$.

Let us introduce mathematical properties of
auxiliary digraphs and describe how $\bbB$ is updated. 
For a pointed partition $\bbB$,
let $H=(\bbB,F)$ be a 1-SR digraph. 
Let us denote by $s$ the number of WCCs in $H$.
For $j\in[1,s]$, we denote by $H_j$ the $j$-th WCC
and by $\bbB_j$ the collection of seed-subset nodes in $H_j$,
that is, $H_j=H[\bbB_j]$ and
$\bbB=\bbB_1\cup\bbB_2\cup\dots\cup\bbB_s$. 
We define a partition $\bbB[F]$ of $\MB^{\propag}_\MS$ by
$\bbB[F]\triangleq\{\cup\bbB_1,\cup\bbB_2,\dots,\cup\bbB_s\}$,
where $\cup\bbB_j$ is the subset of seeds in the
seed-subset nodes in $H_j$. 
Let $z_j\coloneqq|\bbB_j|$. 
We define $\mu(H)$ to be the number of seed-subset nodes in $H$
that belong to functional WCCs, i.e.,
$\mu(H)\triangleq\sum_{j\in[1,s]:\,H_j\textrm{ is functional}} z_j$.
In \figref{1-SR}, for example, 
$\bbB_6=\{\{B_{4,1}\},\{B_{4,2}\},\{B_{5,1}\}\}$ is a set of singletons
whereas $\cup\bbB_6=\{B_{4,1},B_{4,2},B_{5,1}\}$ is the set of seeds
that appear in $\bbB_6$.
We also see that $H_1,\dots,H_5$ are functional,
that $H_6$ is an in-tree and that $\mu(H)=11$ holds. 

In the following lemma,
(ii) says that $\bbB[F]$ is a pointed partition.
The algorithm employs $\bbB[F]$ as the successor
pointed partition of $\bbB$.
(iv) is a useful property for bounding the iteration times
of the algorithm, where (iii) is used
in the proof for (iv). 
(i) will be used in later analyses.

\begin{lem}
  \label{lem:pointed}
  Suppose that we are given
  a monotone system $\MS=\MS(V,\MZ,\omega,\theta)$.
  Let $\bbB = \bbB_1 \cup \bbB_2 \cup \dots \cup \bbB_s$ be a pointed partition of $\MB^{\propag}_\MS$ and
  $H=(\bbB,F)$ be a 1-SR digraph of $\bbB$.  
  \begin{itemize} 
  \item[\rm (i)] For an in-tree WCC of $H$,
    let $\MB$ denote the root seed-subset node.
    Then $\rho(\MB)$ is a MinRS of $V$. 
  \item[\rm (ii)]
    $\bbB[F]$ is a pointed partition of $\MB^{\propag}_\MS$.
  \item[\rm (iii)]
    Let $H_i=H[\bbB_i]$ be an in-tree WCC of $H$. 
    Then $\cup\bbB_i$ has no outgoing arcs
    in the SR digraph of $\bbB[F]$.
  \item[\rm (iv)]
    Let $H'=(\bbB[F],F')$ be a 1-SR digraph of $\bbB[F]$.
    It holds that $\mu(H')\le\mu(H)/2$. 
  \end{itemize}
\end{lem}
\begin{proof}
  (i) Let $B\in\MB^{\propag}_\MS$ be any seed. 
  If $B\in\MB$, then $\rho(B)\supseteq\rho(\MB)$ holds by
  the definition. 
  Otherwise, let $\MB'\in\bbB$ denote the seed-subset
  node such that $B\in\MB'$. 
  Then $\rho(\MB)\not\supseteq\rho(B)$ must hold
  since otherwise $\rho(\MB)\supseteq\rho(B)\supseteq\rho(\MB')$
  would hold,
  contradicting that $H$ contains no arc $(\MB,\MB')$.
  We have seen that there is no seed $B\in\MB^{\propag}_\MS$ that satisfies
  $\rho(\MB)\supsetneq\rho(B)$, and hence
  $\rho(\MB)$ is a MinRS of $V$ by \lemref{minRS_necsuf}. 

  \smallskip
  \noindent(ii)
  We show that $\cup\bbB_j$ is pointed
  for each $j\in[1,s]$, that is,
  there is a seed $B^\ast\in(\cup\bbB_j)$
  such that $\rho(B^\ast)\subseteq\rho(B)$ holds for any $B\in(\cup\bbB_j)$. 
  For $B\in(\cup\bbB_j)$, let $\MB\in\bbB_j$
  denote the seed-subset node such that $B\in\MB$.
  Choose the seed-subset node $\MB^\ast\in\bbB_j$ as follows;
  if $H_j$ is an in-tree, then choose the root seed-subset node
  as $\MB^\ast$, and otherwise (i.e., if $H_j$ is functional),
  choose any seed-subset node on the only cycle as $\MB^\ast$. 
  The chosen seed-subset node $\MB^\ast\in\bbB$
  is pointed, and
  there is $B^\ast\in\Sigma(\MB^\ast)$ such that
  $\rho(B^\ast)=\rho(\MB^\ast)$. 
  Then it holds that $\rho(B)\supseteq\rho(\MB)\supseteq\rho(\MB^\ast)=\rho(B^\ast)$.

  \smallskip
  \noindent(iii)
  Let $\MB^\ast$ denote the root seed-subset node of $H_i$. 
  There is a seed $B^\ast\in\MB^\ast$ such that
  $\rho(\MB^\ast)=\rho(B^\ast)\subseteq\rho(B)$ holds
  for every
  $B\in(\cup\bbB_i)$. 
  Then $B^\ast\in\Sigma(\cup\bbB_i)$ holds and thus
  we have $\rho(B^\ast)=\rho(\cup\bbB_i)$. 
  In the SR digraph of $\bbB[F]$,
  if there is an arc $(\cup\bbB_i,\cup\bbB_j)$
  for some $j\in[1,s]$,
  then 
  $\rho(\cup\bbB_i) \supseteq \rho(\cup\bbB_j)$ holds.
  There would exist $\MB\in\bbB_j$ such that
  $\rho(\MB)=\rho(\cup\bbB_j)\subseteq\rho(\cup\bbB_i)=\rho(\MB^\ast)$, 
  contradicting that the arc $(\MB^\ast,\MB)$ does not
  exist in $H$.
  
  \smallskip
  \noindent(iv) 
  Let $s_1$ (resp., $s_2$)
  denote the number of in-tree (resp., functional) WCCs of $H$, where $s=s_1+s_2$. 
  A functional digraph consists of at least two nodes,
  and thus we have $2s_2\le\mu(H)$.
  In $H'$, there are $s$ seed-subset nodes
  each of which corresponds to a WCC of $H$. 
  By (iii), each in-tree WCC of $H$
  becomes the root seed-subset node
  of an in-tree WCC of $H'$. 
  Then in $H'$, there are at least $s_1$ in-tree WCCs
  and at most $s_2$ seed-subset nodes belong to functional WCCs. 
  Hence we have $\mu(H')\le s_2\le\mu(H)/2$. 
\end{proof}

See \figref{1-SR_r2} for an example of $H'=(\bbB[F],F')$
that is constructed for the 1-SR digraph 
$H=(\bbB,F)$ in \figref{1-SR}.
In this example, we see $s=6$, $s_1=1$, $s_2=5$, $\mu(H)=11$,
and $\mu(H')=2$, where $\mu(H')\le\mu(H)/2$ holds. 

\begin{figure}[t!]
  \centering
  \resizebox{0.9\textwidth}{!}{%

\begin{tikzpicture}[
  >=Latex,
  node distance=1.4cm and 2.0cm,
  box/.style={draw, double, minimum width=8mm, minimum height=8mm, inner sep=0pt},
  lbl/.style={font=\large},
  wlbl/.style={font=\large},
  flow/.style={->, shorten <=2pt, shorten >=2pt}
]

  
\node[box] (a1) at (0,0) {};
\node[box] (a2) at (2.2,0) {};
\node[box] (a3) at (1.1,-1.7) {};
\draw[flow] (a1.south) to[out=-70,in=155] (a3.north west);
\draw[flow] (a2.south) to[out=-110,in=25] (a3.north east);
\node[wlbl, text=blue!70!black, above=4mm] at (-1,0) {$\widehat{\mathbb{B}_1}=\{B_{1,1},B_{1,2}\}$};
\node[wlbl, text=blue!70!black, above=4mm] at (a2.east) {$\widehat{\mathbb{B}_2}=\{B_{1,3},B_{1,4}\}$};
\node[wlbl, text=blue!70!black, below=2mm] at (a3.south) {$\widehat{\mathbb{B}_3}=\{B_{2,1},B_{2,2},B_{2,3}\}$};

\node[box] (b1) at (7.0,0) {};
\node[box] (b2) at (7.0,-1.6) {};
\draw[flow] (b1.east) to[out=-20,in=20] (b2.east);
\draw[flow] (b2.west) to[out=160,in=-160] (b1.west);
\node[wlbl, text=blue!70!black, above=2mm] at (b1.north) {$\widehat{\mathbb{B}_4}=\{B_{3,1},B_{3,2}\}$};
\node[wlbl, text=blue!70!black, below=2mm] at (b2.south) {$\widehat{\mathbb{B}_5}=\{B_{3,3},B_{3,4}\}$};

\node[box] (c1) at (11.5,-0.8) {};
\node[wlbl, text=blue!70!black, below=2mm] at (c1.south) {$\widehat{\mathbb{B}_6}=\{B_{4,1},B_{4,2},B_{5,1}\}$};

\end{tikzpicture}

  }
  \caption{A 1-SR digraph $H'=(\bbB[F],F')$ that is constructed for
  the 1-SR digraph $H=(\bbB,F)$ in \figref{1-SR}. }
  \label{fig:1-SR_r2}
\end{figure}
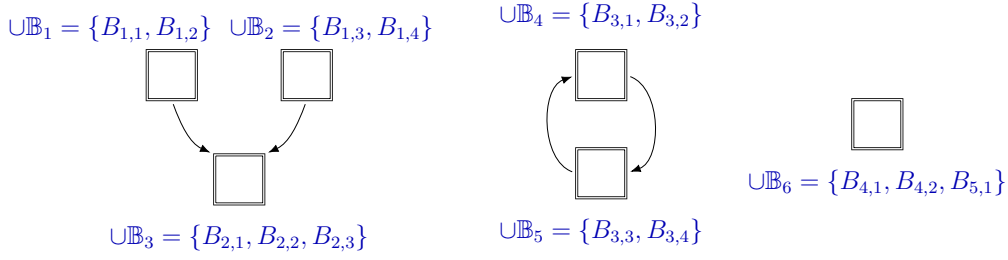

Let $\MY$ denote the set of all MinRSs of $V$.
We can enumerate $\MY$ as follows;
starting from $\bbB=\bbB_\bot$,
we repeat constructing a 1-SR digraph $H=(\bbB,F)$ somehow
and updating $\bbB\coloneqq\bbB[F]$,
where $\bbB[F]$ is pointed by \lemref{pointed}(ii),
until $F=\emptyset$.
The iteration is repeated $O(\log n)$ times
since $\mu(H)$ becomes zero in $O(\log n)$ iterations by \lemref{pointed}(iv),
indicating that all WCCs become in-trees.
This means that all seed-subset nodes become isolated in $O(\log n)$ iterations.
Let $\MB_1,\MB_2,\dots,\MB_t$ denote the isolated
seed-subset nodes of the resulting $H$
and $\MR\coloneqq\{\rho(\MB_j)\mid j\in[1,t]\}$. 
It holds that $\MY\subseteq\MR$ by \lemref{minrs_pointed}
and that $\MY\supseteq\MR$ by \lemref{pointed}(i),
by which $\MY$ is obtained.

\paragraph{How to Construct 1-SR Digraphs.}
We show how to construct a 1-SR digraph $H=(\bbB,F)$ of a pointed partition $\bbB$.
For each seed-subset node $\MB\in\bbB$,
we need to decide whether there is $\MB'$
such that $\rho(\MB)\supseteq\rho(\MB')$ or not;
if such $\MB'$ exists, then we include
the arc $(\MB,\MB')$ into $F$, and otherwise,
we conclude that the seed-subset node $\MB$ has no
outgoing arcs.

The main lemma is \lemref{sigma},
and the preceding Lemmas~\ref{lem:twomin}
and~\ref{lem:twominRS} are preparatory.  

\begin{lem}
  \label{lem:twomin}
  Suppose that we are given
  a monotone system $\MS=\MS(V,\MZ,\omega,\theta)$.
  Let $B,B'\in\MB^{\propag}_{\MS}$ $(B\ne B')$ be seeds.
  The following are equivalent:
  \begin{itemize}
  \item[\rm(i)] 
    There exists
    $u\in\rho(B)\setminus B$ such that $R^{\propag}_\MS[u]\cap B'\ne\emptyset$; and
  \item[\rm(ii)]There exists
    $u\in\rho(B)\setminus B$ such that $N^{\propag+}_\MS[u]\cap B'\ne\emptyset$.
  \end{itemize}
\end{lem}
\begin{proof}
  (ii) $\Longrightarrow$ (i) is trivial
  since $N^{\propag+}_\MS[u]\subseteq R^{\propag}_\MS[u]$ holds.
  For (i) $\Longrightarrow$ (ii),
  by the definition of $R^{\propag}_\MS[u]$,
  there is a path from $u$ to some vertex $v \in R^{\propag}_\MS[u] \cap B'$
  in the propagation digraph $H^{\propag}_\MS$.
  We have $v\in R^{\propag}_\MS[u] \subseteq \rho(u)\subseteq \rho(B)$.
  The vertex $v$ belongs to $B'$
  and $B\cap B'=\emptyset$,
  and hence $v\in\rho(B)\setminus B$.
  It is obvious that $v\in N^{\propag+}_\MS[v]$
  and thus $v\in N^{\propag+}_\MS[v]\cap B'$ holds.
  \invis
  {
  By the monotonicity of $\rho$, we have
  $\rho(B) \supseteq R^{\propag}_\MS[B] \supseteq R^{\propag}_\MS[u]$,
  which implies that $v \in R^{\propag}_\MS[u] \subseteq \rho(B)$.
  Since $B \neq B'$, we have $v \in \rho(B) \setminus B$
  and $v \in N_{\MS}^{\propag+}[v] \cap B' \neq \emptyset$.
  }
\end{proof}

\begin{lem}
  \label{lem:twominRS}
  Suppose that we are given
  a monotone system $\MS=\MS(V,\MZ,\omega,\theta)$.
  Let $B,B'\in\MB^{\propag}_{\MS}$ $(B\ne B')$ be seeds.
  \begin{itemize}
  \item[\rm(i)] 
    If there exists
    $u\in\rho(B)\setminus B$ such that $N^{\propag+}_\MS[u]\cap B'\ne\emptyset$,
    then 
    $\rho(B')\subseteq\rho(B)$ holds.
  \item[\rm(ii)]Otherwise, $\rho(B')\not\subseteq\rho(B)$
    holds.
  \end{itemize}
\end{lem}
\begin{proof}
  (i) If $u\in B'$, then
  we have $\rho(B')=\rho(u)\subseteq\rho(B)$.
  Otherwise, there is an element $v\in N^{\propag+}_{\MS}(u)\cap B'$.
  To obtain $\rho(B)$ by the iterative procedure,
  we need to remove $u$, its neighbor $v$,
  and eventually all elements in $B'$.
  Hence $\rho(B)\supseteq\rho(B')$ holds. 

  \noindent(ii) If $\rho(B)\setminus B=\emptyset$ (i.e., $\rho(B)=B$),
  then
  $\rho(B')\not\subseteq\rho(B)$ holds
  since $B'\subseteq\rho(B')$ and $B=\rho(B)$ are disjoint. 
  Suppose that $N^{\propag+}_\MS[u]\cap B'=\emptyset$ for any $u\in\rho(B)\setminus B$.
  By \lemref{twomin}, $R^{\propag}_\MS[u]\cap B'=\emptyset$ holds
  for any $u\in\rho(B)\setminus B$.
  If $\rho(B) \cap B' \neq \emptyset$,
  then let  $v \in \rho(B) \cap B'$.
  However, by $B\cap B'=\emptyset$,
  we would have $v \in \rho(B) \setminus B$
  and $v \in R^{\propag}_\MS[v] \cap B' \neq \emptyset$,
  a contradiction.
  Then $\rho(B)\cap B'=\emptyset$ holds, 
  implying 
  $\rho(B)\not\supseteq\rho(B')$.
\end{proof}

\begin{lem}
  \label{lem:sigma}
  Suppose that we are given
  a monotone system $\MS=\MS(V,\MZ,\omega,\theta)$.
  For a pointed partition $\bbB$ of $\MB^{\propag}_{\MS}$,
  let $\MB\in\bbB$ be a seed-subset node
  and $B\in\Sigma(\MB)$ be a seed whose closure is minimal
  in $\MB$.  
  \begin{itemize}
  \item[\rm(i)]
    If there are $\MB'\in\bbB$ $(\MB'\ne\MB)$,
    $B'\in\MB'$ and $u\in\rho(B)\setminus B$
    such that 
    $N^{\propag}_\MS[u]\cap B'\ne\emptyset$, 
    then $\rho(\MB')\subseteq\rho(B)$ holds.
  \item[\rm(ii)] Otherwise, $\rho(B)$ is a MinRS of $V$. 
  \end{itemize}
\end{lem}
\begin{proof}
  For (i), we see $B\ne B'$.  By \lemref{twominRS}(i),
  we have  
  $\rho(B)\supseteq\rho(B')\supseteq\rho(\MB')$. 
  For (ii), for any seed $B'\in\MB'$ $(\MB'\ne\MB)$, 
  $\rho(B')\not\subseteq\rho(B)$ holds
  by \lemref{twominRS}(ii).
  For any seed $B'\in\MB$, it holds that $\rho(B')\supseteq\rho(\MB)=\rho(B)$ by the definition of $B$.
  No $\rho(B')$ is a proper subset of $\rho(B)$,
  and by \lemref{minRS_necsuf},
  $\rho(B)$ is a MinRS of $V$. 
\end{proof}

For a pointed partition $\bbB$
and a seed-subset node $\MB\in\bbB$,
we explain how to decide whether
there is another seed-subset node $\MB'\in\bbB$
such that $\rho(\MB)\supseteq\rho(\MB')$.
We exploit \lemref{sigma}(i) and \algref{rho}.
Let $B\in\Sigma(\MB)$. 
Starting from $Y\gets B$, 
we attempt to construct the closure $\rho(B)$ in the manner
of \algref{rho}. During the construction,
we repeatedly add a violating vertex $u\in V\setminus Y$
(i.e., $\omega_{V\setminus Y}(u)\not\geqslant\theta$) to $Y$.
This $u$ must belong to $\rho(B)\setminus B$. 
If there is a seed $B'\in\MB'$ $(\MB'\ne\MB)$
such that $N^{\propag+}_\MS[u]\cap B'\ne\emptyset$,
then $\rho(\MB')\subseteq\rho(B)$ holds by \lemref{sigma}(i),
indicating that the SR digraph includes
the arc $(\MB,\MB')$, and so does a 1-SR digraph. 
On the other hand, if no violating vertex is left
as the result of the repetition,
then $Y=\rho(B)$ holds, that is, the closure is constructed. 
In this case, the seed subset-node $\MB$ has
no outgoing arcs.

For a given seed-subset node $\MB\in\bbB$,
we summarize in \algref{extrho}
the algorithm that identifies whether or not there is $\MB'\ne\MB$ such that
$(\MB,\MB')\in F^\ast[\bbB]$ and outputs the closure
$\rho(\MB)$ if no such $\MB'$ exists. 
Executing \algref{extrho} for each $\MB\in\bbB$,
we can construct a 1-SR digraph $H=(\bbB,F)$,
where $F$ is the set of all $(\MB,\MB')$
that are output by \algref{extrho}. 

\begin{algorithm}[t!]
  \caption{An algorithm for a monotone system to
    identify whether or not there is $\MB'\ne\MB$ such that
    $(\MB,\MB')\in F^\ast[\bbB]$; and outputs the closure
    $\rho(\MB)$ if no such $\MB'$ exists}
  \label{alg:extrho}
  \Input{A monotone system $\MS=\MS(V,\MZ,\omega,\theta)$,
    a pointed partition $\bbB$ of $\MB^{\propag}_\MS$, 
    a seed subset $\MB\in\bbB$, and
    a seed $B\in\Sigma(\MB)$}
  \Output{An ordered pair $(\MB,\MB^\dagger)$ or the closure $\rho(\MB)$, where the former indicates $\rho(\MB)\supseteq\rho(\MB^\dagger)$}
  $Y\gets B$\;
  initialize a set $Q\gets\{v\in V\setminus Y\mid \omega_{V\setminus Y}(v)\not\geqslant\theta\}$\;\label{line:extrho_pre}
  \While{$Q\neq\emptyset$\label{line:extrho_while}}
        {
          $u \gets$ an arbitrary vertex in $Q$\;\label{line:extrho_while_select}
          $Q\gets Q\setminus\{u\}$\;
          \If{there is $\MB'\in\bbB\setminus\{\MB\}$ such that
            $N^{\propag+}_\MS[u]\cap B'\neq\emptyset$ for some $B'\in\MB'$\label{line:extrho_if}}{
            output $(\MB,\MB')$ and halt\;\label{line:extrho_if_output}}
          $Y\gets Y\cup\{u\}$\;\label{line:extrho_insertY}
          \For{each $w \in V \setminus (Y \cup Q)$\label{line:extrho_for}}{
            \If{$\omega_{V\setminus Y}(w)\not\geqslant\theta$}{$Q\gets Q\cup\{w\}$\;}\label{line:extrho_if_Q_update}
          }
        }
        output $Y$ as $\rho(\MB)$\;
\end{algorithm}


\paragraph{Algorithms for Enumerating MinRSs of $V$.}
Now we are ready to present
an algorithm to enumerate all MinRSs of $V$.
Starting from $\bbB=\bbB_\bot$,
the algorithm iteratively constructs a 1-SR digraph $H=(\bbB,F)$
and updates $\bbB\coloneqq\bbB[F]$ until $F=\emptyset$.
Finally we have a pointed partition $\bbB$ such that
the 1-SR digraph consists of isolated seed-subset nodes. 
The algorithm outputs the closure of each node
as a MinRS of $V$.
The preceding discussion guarantees
the correctness of the algorithm.

The algorithm is summarized in \algref{local_minrs}.
In \algref{local_minrs}, for a seed-subset node
$\MB\subseteq\MB^{\propag}_\MS$, 
a variable $\sigma(\MB)$ stores
a seed in $\Sigma(\MB)$; and
a variable $\wasInTree(\MB)$ indicates
whether the seed-subset node $\MB$
belonged to an in-tree of the 1-SR digraph
that was constructed in the last iteration,
where it is initialized for each singleton to $\false$ in \lineref{minrs_init}. 

\begin{algorithm}[t!]
  \caption{An algorithm to compute all MinRSs of $V$}
  \label{alg:local_minrs}
  \Input{A monotone system $\MS=\MS(V,\MZ,\omega,\theta)$ and
    its propagation graph $H^{\propag}_\MS$
    }
  \Output{All MinRSs of $V$}
  $\bbB\gets \bbB_\bot$\; 
  \For{$\MB=\{B\}\in\bbB$}
      {$\sigma(\MB)\gets B$\tcp*{$\sigma(\MB)$ stores a seed in $\Sigma(\MB)$ for \algref{extrho}}
        $\wasInTree(\MB)\gets \false\label{line:minrs_init}$\tcp*{\algref{extrho} is invoked only for $\MB$ such that $\wasInTree(\MB)=\false$}
      }
      
      \Repeat{$F=\emptyset$}{\label{line:local_minrs_repeat_start}
        \tcp{Construct the arc set $F$ for the current $\bbB$}
        $F\gets\emptyset$\;
    \For{$\MB\in\bbB$ with $\wasInTree(\MB)=\false$\label{line:for_MB_local_minrs}}{
      execute \algref{extrho} for $(\MS,\bbB,\MB,\sigma(\MB))$,
      where either $(\MB,\MB^\dagger)$ or $\rho(\MB)$ is returned\label{line:from_minrs_extrho_call}\;
      \lIf{$(\MB,\MB^\dagger)$ is returned}{$F\gets F\cup\{(\MB,\MB^\dagger)\}$}\label{line:local_minrs_returned}
    }
    \tcp{Update $\bbB$ to $\bbB[F]$ if $F\ne\emptyset$}
    \If{$F\ne\emptyset$\label{line:if_F_not_empty_local_minrs}}{
      $H_1,H_2,\dots,H_q\gets$
      WCCs of $H=(\bbB,F)$ such that 
      at least one seed-subset node $\MB$ with $\wasInTree(\MB)=\false$ is contained\label{line:compute_WCCs}\;
      \For{$i=1,2,\dots,q$}{
        $\bbB_i \gets $ the collection
        of seed-subset nodes in $H_i$\;
        \If{$H_i$ is an in-tree}{
          $\MB\gets$ the root seed-subset node of $H_i$\;
          $\wasInTree(\widehat{\bbB_i})\gets \true$\;
        }
        \Else{$\MB\gets$ any seed-subset node in the only cycle of $H_i$\;
        $\wasInTree(\widehat{\bbB_i})\gets \false$\label{line:wasInTree_false}\;}
        $\sigma(\widehat{\bbB_i})\gets \sigma(\MB)$\;
      }
      $\bbB\gets\bbB[F]$\;\label{line:local_minrs_repeat_end}
    }
  }
  output $\{\rho(\MB)\mid \MB\in\bbB\}$ as the set of all MinRSs of $V$\;\label{line:alg_terminate}
\end{algorithm}


In \lineref{from_minrs_extrho_call},
we execute \algref{extrho} to obtain
either $(\MB,\MB^\dagger)$ or $\rho(\MB)$
only for a seed-subset node $\MB\in\bbB$
with $\wasInTree(\MB)=\false$.
This is because we do not need to
do it for $\MB'$ with $\wasInTree(\MB')=\true$
(i.e., $\MB'$ corresponds to
an in-tree of the 1-SR digraph in the last iteration);
by \lemref{pointed}(iii), it is guaranteed
that the seed-subset node $\MB'$ has no outgoing arcs
in the 1-SR digraph that is to be constructed
in the current iteration.
As we will see in the next section,
this leads to an improvement in time complexity. 


For the time complexity,
\algref{local_minrs} runs in $O(n^3\log n\cdot\tau_\omega)$ time
since the repeat-loop from
\lineref{local_minrs_repeat_start} to~\ref{line:local_minrs_repeat_end}
is iterated $O(\log n)$ times by \lemref{pointed}(iv),
and the most time-consuming part, the for-loop from 
\lineref{for_MB_local_minrs} to~\ref{line:local_minrs_returned},
takes $O(n^3\cdot\tau_\omega)$ time; 
\algref{extrho} is called $O(n)$ times
and runs in $O(n^2\cdot\tau_\omega)$ time per each call. 

In the next section, we improve the time complexity
when $\MS$ is local monotone, along with space complexity analyses. 


\section{Local Monotone Systems}
\label{sec:lms}

In this section, we consider the MinRS enumeration
problem on local monotone systems.
In \secref{mmw_improve_eff},
we analyze the computational complexity
of the MinRS enumeration algorithm that
was presented as \algref{local_minrs}
when we are given a local monotone system, 
as a proof for \thmref{main_local_system}.  
In \secref{in-dominating_property},
we improve the complexity when
the local monotone system satisfies
what we call the in-dominating seed property,
as a proof for \thmref{in-dominating}. 
Finally in \secref{examples_local_monotone_systems},
we provide application examples of the framework,
including a proof for \thmref{deg-at-least-k}.

\subsection{Complexity of MinRS Enumeration}
\label{sec:mmw_improve_eff}

For a local monotone system $\MS=\MS(V,\MZ,\omega,\theta;G)$
with an underlying graph $G$,
it is easy to see that the propagation graph
is a subgraph of (oriented) $G$
if $V$ is a $\theta$-solution. 

\begin{lem}
  \label{lem:local}
  Suppose that we are given
  a local monotone system~$\MS(V,\MZ,\omega,\theta;G=(V,E))$
  such that $V$ is a $\theta$-solution.
  Let $(u,v)$ be any arc
  in the propagation digraph $H^{\propag}_\MS$.
  If $G$ is undirected, then
  $uv\in E$ holds. 
  If $G$ is directed, then
  $(u,v)\in E$ and/or $(v,u)\in E$ holds. 
\end{lem}
\begin{proof}
  The whole vertex set $V$ is a $\theta$-solution,
  and hence $\omega_{V}(x)\geqslant \theta$ holds
  for any $x \in V$.
  For an undirected (resp., a directed) $G$,
  suppose that there is no
  edge (resp., arc)
  between $u$ and $v$. 
  We have $\omega_{V\setminus\{u\}}(v) = \omega_{(V\setminus\{u\}) \cap N_G(v)}(v) = \omega_{V \cap N_G(v)}(v) = \omega_V(v) \geqslant \theta$, implying that $(u,v)$ is not an arc in $H^{\propag}_\MS$.
\end{proof}

We can decide
whether $V$ is a $\theta$-solution or not
in $O(n\tau_\omega)$ time and $O(\kappa_\omega)$ space. 
If $V$ is not a $\theta$-solution, then 
we can obtain its unique MinRS by \algref{rho},
which can be done in $O((n+m)\tau_\omega)$ time
and $O(n+m+\kappa_\omega)$ space. 
Otherwise, before executing \algref{local_minrs}, 
we need to construct the propagation digraph $H^{\propag}_\MS$
and then to decompose it into SCCs to obtain seeds.
When $V$ is a $\theta$-solution, 
  let $u,v$ be any vertices in the underlying graph $G=(V,E)$.
  If $G$ is undirected (resp., directed), 
  then by \lemref{local},
  an arc $(u,v)$ exists in $H^{\propag}_\MS$
  only when $uv\in E$ (resp., $(u,v)\in E$ and/or $(v,u)\in E$).
  Hence we can obtain the arc set of $H^{\propag}_\MS$
  by evaluating whether $\omega_{V\setminus\{u\}}(v)\geqslant\theta$ or not for each $u\in V$ and $v\in N_G(u)$.
The propagation graph $H^{\propag}_\MS$
consists of $n$ vertices and $O(n+m)$ arcs. 
We see that construction of $H^{\propag}_\MS$
and its SCC decomposition can be done 
in $O((n+m)\tau_\omega)$ time and $O(n+m+\kappa_\omega)$ space.

We examine the computational complexity of
\algref{extrho} for local monotone systems,
which is repeatedly called as a subroutine
from \algref{local_minrs}. 
\begin{lem}
  \label{lem:1-br_local}
  Given a local monotone system $\MS=\MS(V, \MZ, \omega, \theta; G)$,
  a pointed partition $\bbB$ of $\MB^{\propag}_\MS$,
  a seed-subset node $\MB\in\bbB$ and a seed $B\in\Sigma(\MB)$,
  if $V$ is a $\theta$-solution, then
  \algref{extrho} runs correctly in
  $O((n+m)\tau_\omega)$ time
  and $O(n+m+\kappa_\omega)$ space.
\end{lem}
\begin{proof}
  %
  For the correctness, see the discussion on \algref{extrho} in \secref{mono_alg}. Here we analyze the complexity.
  
In \algref{extrho}, we maintain a subset $Q$ of vertices.
We may assume that it takes $O(1)$ time
to insert a vertex into $Q$ or to delete a vertex from $Q$ 
since $Q$ can be implemented by a queue or a stack.
In \lineref{extrho_pre},
to initialize $Q$, we evaluate $\omega_{V\setminus Y}(u)$
for each $u \in V \setminus Y$,
which can be done in $O(n\tau_\omega)$ time.

We claim that, during the algorithm, 
each vertex $v$ 
is inserted into $Q$ at most once.
It is inserted into $Q$
in Line~\ref{line:extrho_pre} or~\ref{line:extrho_if_Q_update}
when $\omega_{V\setminus Y}(v) \not\geqslant \theta$ holds.
Once $v$ is inserted into $Q$,
it continues to belong to $Q$
until it is picked up as $u$ in \lineref{extrho_while_select}.
After that, \algref{extrho} halts
or it is inserted into $Y$
and is not inserted to $Q$ again, showing the claim. 
We see that a vertex $v$ is inserted into $Q$
when $\omega_{V\setminus Y}(v) \not\geqslant \theta$ holds
for the first time.
This can happen only when a neighbor $u\in N_G(v)$
is added to $Y$ in \lineref{extrho_insertY}
since $\omega$ is local monotone.
Then in \lineref{extrho_for},
it suffices to search 
$N_G(u) \setminus (Y \cup Q)$,
instead of $V \setminus (Y \cup Q)$,
for vertices $w$ such that
$\omega_{V\setminus Y}(w)\not\geqslant\theta$.
The for-loop of \lineref{extrho_for}
takes $O(n+m\tau_\omega)$ time in all. 

%

In \lineref{extrho_if}, we decide 
whether there is $\MB'\in\bbB$ $(\MB'\ne\MB)$
such that $N^{\propag+}_\MS[u]\cap B'\neq\emptyset$ for some $B'\in\MB'$, which can be done in $O(\deg_G(u))$ time. 
Summing over all vertices, the total time 
is $O(n+m)$.
%
We see that \algref{extrho} runs in $O((n+m)\tau_\omega)$ time.


For the space complexity,
we use $O(n+m)$ space to store the underlying graph $G$,
the propagation digraph $H^{\propag}_\MS$
and sets such as $Y$ and $Q$.
We also use $\kappa_\omega$ space
to compute $\omega_S(u)$
for a vertex $u$ and a subset $S\subseteq V$.
Thus, \algref{extrho} runs in $O(n+m+\kappa_\omega)$ space.
\end{proof}

\begin{lem}
  \label{lem:local_minrs}
  Given a local monotone system $\MS=\MS(V, \MZ, \omega, \theta; G)$ and the propagation graph $H^{\propag}_\MS$,
  if $V$ is a $\theta$-solution, then
  \algref{local_minrs} runs correctly in
  $O((n+m)n\tau_\omega)$ time
  and $O(n+m+\kappa_\omega)$ space.
\end{lem}
\begin{proof}
  For the correctness, see the discussion on \algref{local_minrs} in \secref{mono_alg}. Here we analyze the complexity.
  
  The most time-consuming part is the for-loop
  from \lineref{for_MB_local_minrs} to~\ref{line:local_minrs_returned}.
  A rough analysis may be as follows;
  the for-loop takes $O(|\bbB|(n+m)\tau_\omega)=O((n+m)n\tau_\omega)$ time since \algref{extrho} runs in $O((n+m)\tau_\omega)$ time by
  \lemref{1-br_local}.
  The repeat-loop from \lineref{local_minrs_repeat_start} to~\ref{line:local_minrs_repeat_end}
  iterates $O(\log n)$ times by \lemref{pointed}(iv),
  and the overall time complexity is 
  $O((n+m)n\log n\cdot\tau_\omega)$.


  We claim that the $O(\log n)$ factor
  can be dropped from the above time bound.
  Observe that, in the for-loop
  from \lineref{for_MB_local_minrs}
  to~\ref{line:local_minrs_returned},
  \algref{extrho} is executed only for
  a seed-subset node $\MB\in\bbB$
  with $\wasInTree(\MB)=\false$.
  Let $s$ denote 
  the number of seed-subset nodes
  $\MB$ with $\wasInTree(\MB)=\false$.
  We can implement the for-loop
  so that it iterates precisely $s$ times;
  e.g., by maintaining the $s$ seed-subset nodes by queue.
  
  It holds that $s=n$ 
  in the first iteration of the repeat-loop,
  and in subsequent iterations,
  $s$ is the number of functional WCCs
  in the 1-SR digraph of the last iteration.
  Similarly to \lemref{pointed}(iv),
  we can show that the number $s$ is at most halved
  in the next iteration.
  Therefore, the for-loop
  from \lineref{for_MB_local_minrs}
  to~\ref{line:local_minrs_returned} iterates
  at most $n+\frac{1}{2}n+(\frac{1}{2})^2n+\dots=O(n)$ times,
  indicating that \algref{extrho} is called
  $O(n)$ times over an execution of \algref{local_minrs}, showing the claim. 
  
  In \lineref{if_F_not_empty_local_minrs},
  we can compute required WCCs of $H=(\bbB,F)$
  by running a graph search algorithm
  from each seed-subset node $\MB\in\bbB$
  with $\wasInTree(\MB)=\false$.
  The size of WCCs is $O(s)$ and 
  this part takes $O(s)$ time per iteration
  of the repeat-loop,
  and thus $O(n)$ time over the algorithm. 
  One can readily see that the time complexity of the remaining parts
  is $O((n+m)n\tau_\omega)$.

For the space complexity,
we use $O(n+m)$ space to store the underlying graph $G$,
the propagation digraph $H^{\propag}_\MS$
and sets such as $Y$ and $Q$ for \algref{extrho},
and we use $O(n)$ space to store the 1-SR digraph $H=(\bbB,F)$ and the pointed partition $\bbB$.
In addition,
we use $O(n + m + \kappa_\omega)$ space
in each execution of \algref{extrho}.
Thus, the overall space complexity is $O(n+m+\kappa_\omega)$.
\end{proof}

\paragraph{Proof of \thmref{main_local_system}.}
Based on the discussion at the beginning of this subsection,
it suffices to consider the case
when $V$ is a $\theta$-solution.
We can construct the propagation graph
and compute the SCC decomposition in 
$O((n+m)\tau_\omega)$ time and $O(n+m+\kappa_\omega)$ space.
By \lemref{local_minrs},
\algref{local_minrs} outputs all MinRSs of $V$ in
$O((n+m)n\tau_\omega)$ time and $O(n+m+\kappa_\omega)$ space,
showing \thmref{main_local_system}. 
\myqed

\subsection{Efficiency Enabled by In-Dominating Seed Property}\label{sec:in-dominating_property}



For a digraph $D=(V,A)$,
a subset $S \subseteq V$ is an \emph{in-dominating set} in $D$
if, for every vertex $u \in V \setminus S$,
there is a vertex $v \in S$ such that $(u,v) \in A$,
i.e., $N^+_D(u) \cap S \neq \emptyset$ holds.
We say that a monotone system $\MS=\MS(V,\MZ,\omega,\theta)$, 
not necessarily local,
satisfies the \emph{in-dominating seed property}
if the set $\widehat{\MB^\propag_\MS}$ of seed vertices
is an in-dominating set
in the propagation graph $H^{\propag}_\MS$.

Let $\MB\subseteq \MB^\propag_\MS$ be a subset of seeds
in the propagation digraph $H^{\propag}_\MS$.
We define $C_\MB$ to be the set of elements $u\in V$
such that, for any closed out-neighbor $v\in N^{\propag+}_\MS[u]$
in the propagation digraph,
$v$ belongs to $\cup\MB$ whenever $v$ is a seed vertex,
that is, 
\begin{align}
  C_\MB
  &\triangleq
  \{u \in V \mid
  N^{\propag+}_\MS[u] \cap (\widehat{\MB^\propag_\MS})
  \subseteq
  (\widehat{\MB})
  \}
  = \{u \in V \mid
  N^{\propag+}_\MS[u] \cap (\widehat{\MB^\propag_\MS}\setminus \widehat{\MB})
  = \emptyset
  \}
  \nonumber\\
  &= \{ u \in V \mid
  \text{\(N^{\propag+}_\MS[u] \cap (\widehat{\MB'}) = \emptyset \) for any \(\MB' \in \bbB \setminus \{\MB\}\)} \}
  \nonumber\\
  &= \{ u \in V \mid
  \text{\(N^{\propag+}_\MS[u] \cap B' = \emptyset\)
  for any \(B'\in\MB'\) with \(\MB'\neq\MB\)} \}.\label{eq:C_MB}
\end{align}
Obviously we have $(\cup\MB)\subseteq C_\MB$,
and $C_\MB\setminus(\cup\MB)$ is the set of
elements $u\in V$ such that every out-neighbor of $u$
that is a seed vertex 
belongs to a seed in $\MB$. 

Given a monotone system $\MS$,
let $\bbB$ be a partition of the seed set $\MB^{\propag}_\MS$. 
The following lemma says that,
if $\MS$ satisfies the
in-dominating seed property, then
$C_\MB\cap C_{\MB'}=\emptyset$ holds for
any $\MB,\MB'\in\bbB$. 
\begin{lem}
  \label{lem:C_disjoint}
  For a monotone system $\MS=\MS(V,\MZ,\omega,\theta)$,
  let $\MB,\MB'\subseteq\MB^{\propag}_\MS$ be subsets of seeds
  such that $\MB\cap\MB'=\emptyset$.
  If $\MS$ satisfies the in-dominating seed property,
  then it holds that $C_\MB\cap C_{\MB'}=\emptyset$. 
\end{lem}
\begin{proof}
  Suppose $C_\MB\cap C_{\MB'}\ne\emptyset$ for contradiction.
  Let $v\in C_\MB\cap C_{\MB'}$.
  By the definition of $C_\MB$ and $C_{\MB'}$,
  we have $N^{\propag+}_\MS[v] \cap (\widehat{\MB^\propag_\MS})\subseteq ((\cup\MB)\cap(\cup\MB'))=\emptyset$,
  where the last equality holds by $\MB\cap\MB'=\emptyset$.
  We claim that
  $v$ is not a seed vertex (i.e., $v\in(\cup\MB^\propag_\MS)$),
  since otherwise
  $N^{\propag+}_\MS[v] \cap (\widehat{\MB^\propag_\MS})$
  would not be empty.
  In the propagation digraph, by $N^{\propag+}_\MS[v] \cap (\widehat{\MB^\propag_\MS})=\emptyset$,
  no out-neighbor of $v$ would be a seed vertex,
  indicating that $\cup\MB^{\propag}_\MS$
  is not in-dominating, a contradiction. 
\end{proof}

Observe that, if $\MS$ does not satisfy the in-dominating seed property, then $C_\MB\cap C_{\MB'}$ is not empty;
there is a vertex $v\in V$ such that no out-neighbor
is a seed vertex, and $v$ belongs
to both $C_\MB$ and $C_{\MB'}$ by the definition.

\paragraph{Proof of \thmref{in-dominating}.}
In \algref{local_minrs}, 
recall that the repeat-loop from
\lineref{local_minrs_repeat_start}
to~\ref{line:local_minrs_repeat_end}
iterates $O(\log n)$ times by \lemref{pointed}(iv).
We show that the most time-consuming part,
the for-loop
from \lineref{for_MB_local_minrs} to~\ref{line:local_minrs_returned},
can be done in $O(n+m)$ time per iteration of the repeat-loop. 

Let $\bbB=\{\MB_1,\MB_2,\dots,\MB_z\}$ denote the
pointed partition at the beginning of any iteration
of the repeat-loop. 
For $v\in V$, let $\ell(v)\coloneqq\{i\in[1,z]\mid N^{\propag+}_\MS[v]\cap(\cup\MB_i)\ne\emptyset\}$.
The system $\MS$ satisfies the in-dominating property
and hence $|\ell(v)|\ge1$ holds for all $v\in V$,
and $v\in C_{\MB_i}$ holds
when and only when $\ell(v)=\{i\}$.
We compute $\ell(v)$ for all $v\in V$
as a preprocessing of the for-loop,
which can be done in $O(n+m)$ time.

In each iteration of the for-loop, 
we execute \algref{extrho} for each $\MB_i\in\bbB$
such that $\wasInTree(\MB_i)=\false$.
We claim that
the time complexity is $O(|C_{\MB_i}|+|E[C_{\MB_i}]|+|E(C_{\MB_i})|)$, where $E[C_{\MB_i}]$ (resp., $E(C_{\MB_i})$)
is the subset of edges in the underlying graph $G=(V,E)$
such that both endpoints belong to $C_{\MB_i}$
(resp., exactly one of the endpoints belong
to $C_{\MB_i}$).
To show the claim, in \algref{extrho}, 
recall that each vertex is inserted into $Q$
at most once; see the proof for \lemref{1-br_local}.
Let $u$ denote any vertex chosen
in \lineref{extrho_while_select}. 
The algorithm terminates 
if the condition in \lineref{extrho_if}
is satisfied, that is,
there are $\MB_j\in\bbB$ $(\MB_j\ne\MB_i)$ and
$B'\in\MB_j$ such that
$N^{\propag+}_\MS[u]\cap B'\ne\emptyset$.
In this case, $u\not\in C_{\MB_i}$ holds by~\eqref{C_MB}.
We can decide whether $u\in C_{\MB_i}$ or not in constant time
using $\ell(u)$; if $\ell(u)=\{i\}$,
then $u\in C_{\MB_i}$ holds,
and otherwise, $u\notin C_{\MB_i}$ holds. 
We can run \algref{extrho} runs in linear time
with respect to $|C_{\MB_i}|$ and
the number of arcs incident to $C_{\MB_i}$,
that is $|E[C_{\MB_i}]|+|E(C_{\MB_i})|$.

Finally we see that $\sum_{i=1}^zO(|C_{\MB_i}|+|E[C_{\MB_i}]|+|E(C_{\MB_i})|)=O(n+m)$ since $C_{\MB_i}\cap C_{\MB_j}=\emptyset$, $i,j\in[1,z]$ $(i\ne j)$ by \lemref{C_disjoint},
as required. 

The space complexity follows by \thmref{main_local_system}. 
\myqed


\subsection{Application Examples}
\label{sec:examples_local_monotone_systems}
We show that the MinRSs of a $k$-core $G=(V,E)$
can be computed in $O((n+m)\log n)$ time,
as a proof for \thmref{deg-at-least-k}.
For this, we prove that the local monotone system
that is defined somehow based on $k$-cores 
satisfies the in-dominating seed property.
Then we can apply \thmref{in-dominating} to
the system. 

There are a lot of extensions of $k$-cores~\cite{MGPV.2020}.
It is easy to see that each extension
determines a similar local monotone system,
and that, by \thmref{main_local_system}, 
we can enumerate all MinRSs of
an extended $k$-core in $O((n+m)n)$ time.
So far we have not found any extensions
such that the bounds can be improved
by using \thmref{in-dominating}; 
in these cases, the in-dominating seed property
does not hold.

This paper has focused on enumerating MinRSs of a ground set $V$
of a local monotone system $\MS(V,\MZ,\omega,\theta;G)$. 
We show that, using the notion of SD property~\cite{TH.2025},
we can enumerate all $\theta$-solutions efficiently.
In particular, for standard $k$-cores in undirected graphs,
the resulting algorithm achieves a better
time complexity than Boley et al.'s framework~\cite{BHPW.2010}. 

\subsubsection{Enumeration of MinRSs of \texorpdfstring{$k$}{k}-Cores}
For a positive integer $k$,
let $G=(V,E)$ be a $k$-core graph;
$\MZ=(\bbZ,\le)$; and
$\omega$ denote the function
such that $\omega_S(u)=\deg_{G[S\cup\{u\}]}(u)$
for $S\subseteq V$ and $u\in V$.
We see that $\omega$ is local monotone.
We denote $\MS=\MS(V,\MZ,\omega,k;G)$,
which is a local monotone system
such that $S\subseteq V$ is a $k$-solution
if and only if $S$ induces a $k$-core graph.

We show that $\MS$ satisfies the in-dominating seed
property. Let us examine the structure
of the propagation graph $H^{\propag}_\MS = (V, A^{\propag}_\MS)$, where 
the arc set 
is given by
$A^{\propag}_\MS =\{(u,v)\in V\times V\mid \deg_{G - u}(v) < k\}$. 
The following lemma is obvious and
will be used soon. 
\begin{lem}
  \label{lem:deg-at-least-k_degree_condition}
  For a positive integer $k$,
  let $G=(V,E)$ be a $k$-core graph. 
  For any distinct vertices $u,v \in V$,
  $\deg_{G - u}(v) < k$ holds
  if and only if
  $\deg_G(v) = k$ and $u \in N_G(v)$.
\end{lem}
\begin{proof}
  The sufficiency is obvious. 
  For the necessity,
  we see that $\deg_G(v)\ge k$ holds since $G$ is a $k$-core graph. 
  It holds that $u\in N_G(v)$ by
  $\deg_{G-u}(v)<k\le\deg_G(v)$.
  Further, $\deg_G(v)=k$ must hold
  by $\deg_{G-u}(v)=\deg_G(v)-1$. 
\end{proof}

We partition the vertex set $V$ of a $k$-core graph
$G=(V,E)$ into $V=V_1\sqcup V_2\sqcup V_3$
such that
$V_1 := \{ v \in V \mid \deg_G(v) > k \text{ and } \deg_G(u) > k \text{ for any } u \in N_G(v) \}$;
$V_2 := \{ v \in V \mid \deg_G(v) > k \text{ and } \deg_G(u) = k \text{ for some } u \in N_G(v) \}$; and
$V_3 := \{ v \in V \mid \deg_G(v) = k \}$.
The following lemma analyzes
the structure of the propagation graph
$H^{\propag}_\MS = (V, A^{\propag}_\MS)$.
In particular, (v) says that
$\MS$ satisfies the in-dominating seed property.

\begin{lem}
  \label{lem:deg-at-least-k_propagation_graph}
  For a positive integer $k$,
  let $G=(V,E)$ be a $k$-core graph.
\begin{itemize}
\item[\rm(i)]
A vertex $v$ has an incoming arc in $H^{\propag}_\MS$
if and only if $v\in V_3$.
\item[\rm(ii)]
Each vertex $v\in V_1$ is isolated in $H^{\propag}_\MS$;
in particular, $\{ v \}$ is a seed.
\item[\rm(iii)]
No vertex in $V_2$ belongs to a bottom SCC of $H^{\propag}_\MS$.
\item[\rm(iv)]
It holds that
$\widehat{\MB^{\propag}_\MS}=V_1\cup V_3$.
\item[\rm(v)]
$\widehat{\MB^{\propag}_\MS}$ is an in-dominating set in
$H^{\propag}_\MS$.
\end{itemize}
\end{lem}

\begin{proof}
  By \lemref{deg-at-least-k_degree_condition},
  the arc set of the propagation graph
  $H^{\propag}_\MS$ is $A^{\propag}_\MS=\{(u,v)\in V\times V
\mid \deg_G(v)=k\ \text{and}\ u\in N_G(v)\}$.

\noindent
(i) If there is $(u,v)\in A^{\propag}_\MS$, then
$\deg_G(v)=k$ holds by the definition,
and hence $v\in V_3$.
If $v\in V_3$, then $\deg_G(v)=k\geq 1$
and hence $v$ has a neighbor $u$ in $G$.
Then $(u,v)\in A^{\propag}_\MS$ holds.

\noindent
(ii)
Let $v\in V_1$.
Since every neighbor of $v$ has degree at least $k+1$,
no arc leaves $v$.
Moreover, by (i), no arc enters $v$. Hence $v$ is isolated in $H^{\propag}_\MS$,
and $\{ v \}$ is a seed.

\noindent
(iii)
Let $v\in V_2$.
By definition, $v$ has a neighbor $u\in V_3$,
and hence $(v,u)\in A^{\propag}_\MS$.
On the other hand, by (i), no arc enters $v$. Therefore $v$ cannot belong to any seed. 

\noindent


\noindent
(iv)
By (ii), $\{u\}$ is a seed for each $u\in V_1$. 
By (iii), no vertex in $V_2$ belongs to a seed.
For $u\in V_3$, if $u$ has no outgoing arcs
in $H^{\propag}_\MS$, then $\{u\}$ is a seed.
Otherwise, there is an arc $(u,v)$ in $H^{\propag}_\MS$,
and $v\in V_3$ holds by (i).
This indicates the arc $(v,u)$ also exists
in $H^{\propag}_\MS$, and hence $u$ and $v$
belong to the same SCC. This SCC should be a seed
since no arc from $V_3$ to $V_j$, $j\in[1,2]$
can exist by (i). 

\noindent
(v)
For $u\in V\setminus (V_1\cup V_3)=V_2$,
there is an arc $(u,v)\in A^{\propag}_\MS$
such that $v\in V_3$. 
We see that $v\in V_3\subseteq (\widehat{\MB^{\propag}_\MS})$ holds, where the latter relationship is
due to (iv). 
\end{proof}


\paragraph{Proof of \thmref{deg-at-least-k}.}

By \lemref{deg-at-least-k_propagation_graph}(v),
the local monotone system $\MS$ satisfies
the in-dominating seed property.
By \thmref{in-dominating},
we can enumerate all MinRSs of $G$
in $O((n+m)\log n\cdot\tau_\omega)$ time
and $O(n+m+\kappa_\omega)$ space. 


We show that
we can implement the algorithm so that
$\tau_\omega=O(1)$ and $\kappa_\omega=O(n+m)$.
Observe that, in the MinRS enumeration algorithm,
we compute $\omega_S(v)=\deg_S(v)$
for $(S,v)\in 2^V\times V$
in Lines~\ref{line:extrho_pre} and~\ref{line:extrho_if_Q_update} of 
\algref{extrho}. Specifically,
we compute $\omega_{V\setminus Y}(v)$
for a subset $Y\subseteq V$ of vertices,
and the set $Y$ is non-decreasing
with respect to set-inclusion, during
a single execution of \algref{extrho}.

During \algref{extrho}, we maintain
the degree $\omega_{V\setminus Y}(v)$ for
the present $Y\subseteq V$ and each $v\in V$
in an array.
Clearly,
if a vertex $u$ is added to $Y$,
then we decrease by one
the entry of each neighbor $v\in N_G(u)$,
which can be done in $O(1)$ time per $v$,
indicating $\tau_\omega=O(1)$. 
This array is first initialized
in the preprocessing of \algref{local_minrs}. 
\algref{extrho} may update the array
during its execution, and when
it terminates, 
we restore the array in linear time
with respect to the number of visited vertices.
It is easy to see $\kappa_\omega=O(n+m)$.
The overall complexity  
is not affected in the big-O notation. 
\myqed

\subsubsection{Enumeration of MinRSs of Extended \texorpdfstring{$k$}{k}-Cores}
Our framework is applicable to extensions of $k$-cores. 
We take up weighted $k$-cores~\cite{GSH.2012};
multi-layer $\bm{k}$-cores~\cite{GBG.2017}; and
$(k,\ell)$-cores~\cite{GTV.2013} here.
See~\cite{MGPV.2020} for other examples.

\paragraph{Weighted $k$-Cores~\cite{GSH.2012}.}
We denote by $\bbR_{\ge 0}$ the set of non-negative reals. 
Let $\MN=[G,w]$ be a network that consists of
an undirected graph $G=(V,E)$
and an edge-weight function $w:E\to\bbR_{\ge 0}$.
The \emph{weighted degree of a vertex $v\in V$}
is defined to be $\sum_{u\in N_G(v)} w(uv)$.
Let $k\in\bbR_{\ge 0}$ be a constant.
A subset $S\subseteq V$ is the \emph{weighted $k$-core} (\emph{of $\MN$})
if it is the maximal subset such that
every weighted degree in $G[S]$ is no less than $k$. 
The graph $G$ is a \emph{weighted $k$-core graph}
if $G$ itself is the weighted $k$-core of $G$.

Let $\MZ=(\bbR_{\ge0},\le)$; and
$\omega:2^V\times V\to\bbR_{\ge0}$ denote the function such that
$\omega_S(v)=\sum_{u\in N_{S\cup\{v\}}(v)} w(uv)$, $S\subseteq V$, $v\in V$.
One readily sees that $\omega$ is local monotone.
The system $\MS=\MS(V,\MZ,\omega,k;\MN)$ is local monotone such that
$S\subseteq V$ is a $k$-solution if and only if
it induces a weighted $k$-core graph. 

We can enumerate MinRSs of a weighted $k$-core in $O((n+m)n)$ time and $O(n+m)$ space
by \thmref{main_local_system},
where we can implement the algorithm
so that $\tau_\omega=O(1)$ and $\kappa_\omega=O(n+m)$
hold, in an analogous manner with
(unweighted) $k$-cores; see the proof of \thmref{deg-at-least-k}. 

The system does not satisfy the in-dominating
seed property in general. 
See \figref{weight-and-sumdegree_counterexample}.
Consider a local monotone system
$\MS=\MS(V,\MZ,\omega,k;\MN)$ such that
$k=2$ and $\MN=[G,w]$ is a network
that corresponds to an edge-weighted graph in (i).
The propagation digraph is shown in (ii).
We see that the in-dominating seed property
does not hold;
the seed set is $\MB^{\propag}_\MS=\{\{v_1\},\{v_3\},\{v_5\}\}$, and hence $\cup\MB^{\propag}_\MS=\{v_1,v_3,v_5\}$.
This set is not in-dominating in the propagation digraph
since there is no outgoing arc from $v_2\notin(\cup\MB^{\propag}_\MS)$
to any of $\cup\MB^{\propag}_\MS=\{v_1,v_3,v_5\}$.

  
\begin{figure}[t!]
  \centering
  \begin{tabular}{cc}
  \begin{tikzpicture}[scale=0.6,>=Latex,
    flow/.style={->, shorten <=2pt, shorten >=2pt}]
    \tikzmath{
      \a = 3.5;
    }
    \coordinate (v1) at (\a/2, -\a/2) node at (v1) [label=below:$v_1$] {};
    \coordinate (v2) at (0, 0) node at (v2) [label=left:$v_2$] {};
    \coordinate (v3) at (\a, 0) node at (v3) [label=right:$v_3$] {};
    \coordinate (v4) at (0, \a) node at (v4) [label=left:$v_4$] {};
    \coordinate (v5) at (\a, \a)  node at (v5) [label=right:$v_5$] {};
    \coordinate (v6) at (0,2*\a) node at (v6) [label=left:$v_6$] {};
    \coordinate (v7) at (\a,2*\a) node at (v7) [label=right:$v_7$] {};
    \draw (v1)--node [midway, left] {$2$}(v2);
    \draw (v1)--node [midway, right] {$2$}(v3);
    \draw (v2)--node [midway, auto] {$2$}(v3);
    \draw (v2)--node [midway, auto] {$2$}(v4);
    \draw (v4)--node [midway, auto] {$1$}(v5);
    \draw (v5)--node [midway, auto] {$1$}(v7);
    \draw (v7)--node [midway, auto] {$2$}(v6);
    \fill[fill=white, draw] (v1) circle (3pt) (v2) circle (3pt);
    \fill[fill=white, draw] (v3) circle (3pt) (v4) circle (3pt);
    \fill[fill=white, draw] (v5) circle (3pt) (v6) circle (3pt);
    \fill[fill=white, draw] (v7) circle (3pt);
  \end{tikzpicture}
  &
  \begin{tikzpicture}[scale=0.6,>=Latex,
    flow/.style={->, shorten <=2pt, shorten >=2pt}]
    \tikzmath{
      \a = 3.5;
    }
    \coordinate (v1) at (\a/2, -\a/2) node at (v1) [label=below:$v_1$] {};
    \coordinate (v2) at (0, 0) node at (v2) [label=left:$v_2$] {};
    \coordinate (v3) at (\a, 0) node at (v3) [label=right:$v_3$] {};
    \coordinate (v4) at (0, \a) node at (v4) [label=left:$v_4$] {};
    \coordinate (v5) at (\a, \a)  node at (v5) [label=right:$v_5$] {};
    \coordinate (v6) at (0,2*\a) node at (v6) [label=left:$v_6$] {};
    \coordinate (v7) at (\a,2*\a) node at (v7) [label=right:$v_7$] {};
    \draw[flow] (v2) to (v4);
    \draw[flow] (v6) to[out=10,in=170] (v7);
    \draw[flow] (v7) to[out=190,in=350] (v6);
    \draw[flow] (v4) to (v5);
    \draw[flow] (v7) to (v5);
    \fill[fill=white, draw] (v1) circle (3pt) (v2) circle (3pt);
    \fill[fill=white, draw] (v3) circle (3pt) (v4) circle (3pt);
    \fill[fill=white, draw] (v5) circle (3pt) (v6) circle (3pt);
    \fill[fill=white, draw] (v7) circle (3pt);
  \end{tikzpicture}
  \\
  (i) An edge-weighted graph
  & (ii) The propagation digraph 
  \end{tabular}
  \caption{A counterexample
    such that the local monotone system
    based on weighted $k$-cores
    does not satisfy the in-dominating seed property.}
  \label{fig:weight-and-sumdegree_counterexample}
\end{figure}
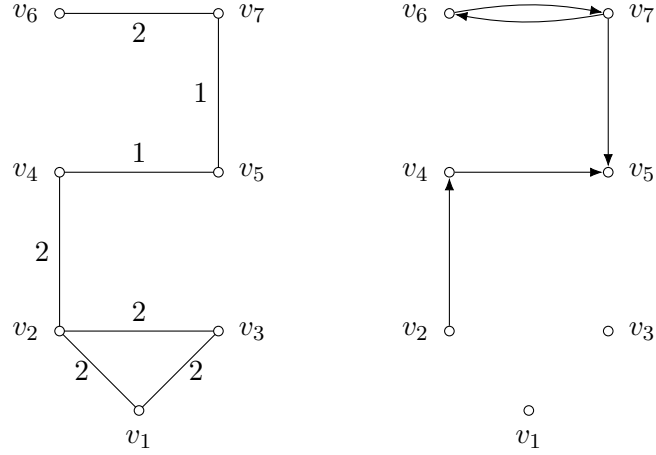

\paragraph{Multi-layer $\bm{k}$-Cores~\cite{GBG.2017}.}
Let $r$ be a positive integer and
$G_i=(V,E_i)$, $i\in[1,r]$ be graphs 
such that $V$ is the common vertex set.
For $r$-dimensional vectors $\bm{x},\bm{y}\in\bbR^r$,
we write $\bm{x}\leqslant\bm{y}$ if $x_i\le y_i$ holds for each $i\in[1,r]$,
where we denote $\bm{x}=(x_1,x_2,\dots,x_r)$ and $\bm{y}=(y_1,y_2,\dots,y_r)$. 
Obviously $\leqslant$ is a partial order on $\bbR^r$.
We define a function $\omega:2^V\times V\to\bbR^r$ to be
$\omega_S(v)\coloneqq(\deg_{G_1[S]}(v),\deg_{G_2[S]}(v),\dots,\deg_{G_r[S]}(v))$,
$S\subseteq V$, $v\in V$. 
For $\bm{k}\in\bbR^r$, 
a subset $S\subseteq V$ is the \emph{multi-layer $\bm{k}$-core}
(\emph{of $\MG=\{G_1,G_2,\dots,G_r\}$})
if it is the maximal subset such that
$\omega_S(v)\geqslant\bm{k}$ holds for every $v\in S$. 
The definition of a multi-layer $\bm{k}$-core graph is analogous. 

Let $\MZ=(\bbR^r,\leqslant)$ denote a poset. 
The set system $\MS=\MS(V,\MZ,\omega,\bm{k};\MG)$ is local monotone
since the function $\omega$ is local monotone. 
A subset $S\subseteq V$ is a $\bm{k}$-solution
if and only if it induces a multi-layer $\bm{k}$-core graph. 
It is easy to show that
the arc set of the propagation digraph $H^{\propag}_\MS$
is the union of those of the analogous propagation digraphs for $G_1,G_2,\dots,G_r$. 
By \thmref{main_local_system}, we can enumerate MinRSs of $V$ in 
$O((n+M)n)$ time and $O(n+M)$ space,
where we denote $M=|E_1|+|E_2|+\dots+|E_r|$.

We cannot utilize \thmref{in-dominating} to improve the time bound
since the in-dominating seed property does not hold in general;
See \figref{multi-layer_counterexample}. 
In this example, we use $V=\{v_1,v_2,\dots,v_{11}\}$, $r=2$ and $\bm{k}=(2,2)$.
The two underlying graphs $G_1$ and $G_2$
are shown in (i) and (ii), respectively.
For the local monotone system
$\MS=\MS(V,\MZ,\omega,\bm{k};\MG=\{G_1,G_2\})$,
the propagation digraph $H^{\propag}_\MS$ is shown in (iii).
Observe that the arc set of $H^{\propag}_\MS$ is the union
of the arc sets of the propagation graphs
that are constructed independently for $G_1$ and $G_2$. 
The seed set of $H^{\propag}_\MS$ is $\MB^{\propag}_\MS=\{\{v_5\},\{v_6\},\{v_7\},\{v_8\},\{v_9,v_{10},v_{12}\}\}$, and $v_4$ has no outgoing arc to any of them.
Then $\cup\MB^{\propag}_\MS=\{v_5,v_6,\dots,v_{10},v_{12}\}$ is not an in-dominating set in $H^{\propag}_\MS$. 

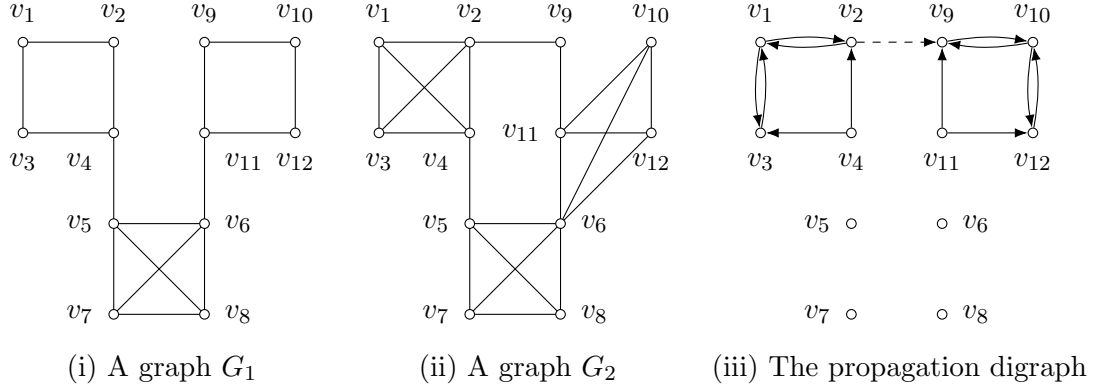
\begin{figure}[t!]
  \centering
  \begin{tabular}{ccc}
  \begin{tikzpicture}[scale=0.6,>=Latex,
    flow/.style={->, shorten <=2pt, shorten >=2pt}]
    \tikzmath{
      \a = 2;
    }
    \coordinate (v1) at (0, 0) node at (v1) [label=above:$v_{1}$] {};
    \coordinate (v2) at (\a, 0) node at (v2) [label=above:$v_{2}$] {};
    \coordinate (v3) at (0, -\a) node at (v3) [label=below:$v_{3}$] {};
    \coordinate (v4) at (\a, -\a) node at (v4) [label=below left:$v_{4}$] {};
    \coordinate (v5) at (\a, -2*\a)  node at (v5) [label=left:$v_{5}$] {};
    \coordinate (v6) at (2*\a,-2*\a) node at (v6) [label=right:$v_{6}$] {};
    \coordinate (v7) at (\a,-3*\a) node at (v7) [label=left:$v_{7}$] {};
    \coordinate (v8) at (2*\a,-3*\a) node at (v8) [label=right:$v_{8}$] {};
    \coordinate (v9) at (2*\a,0) node at (v9) [label=above:$v_{9}$] {};
    \coordinate (v10) at (3*\a,0) node at (v10) [label=above:$v_{10}$] {};
    \coordinate (v11) at (2*\a,-\a) node at (v11) [label=below right:$v_{11}$] {};
    \coordinate (v12) at (3*\a,-\a) node at (v12) [label=below:$v_{12}$] {};
    \draw (v1)--(v2);
    \draw (v1)--(v3);
    \draw (v2)--(v4);
    \draw (v3)--(v4);
    \draw (v4)--(v5);
    \draw (v5)--(v6);
    \draw (v5)--(v7);
    \draw (v5)--(v8);
    \draw (v6)--(v7);
    \draw (v6)--(v8);
    \draw (v7)--(v8);
    \draw (v6)--(v11);
    \draw (v11)--(v9);
    \draw (v11)--(v12);
    \draw (v9)--(v10);
    \draw (v10)--(v12);
    \fill[fill=white, draw] (v1) circle (3pt) (v2) circle (3pt);
    \fill[fill=white, draw] (v3) circle (3pt) (v4) circle (3pt);
    \fill[fill=white, draw] (v5) circle (3pt) (v6) circle (3pt);
    \fill[fill=white, draw] (v7) circle (3pt) (v8) circle (3pt);
    \fill[fill=white, draw] (v9) circle (3pt) (v10) circle (3pt);
    \fill[fill=white, draw] (v11) circle (3pt) (v12) circle (3pt);
  \end{tikzpicture}
  &\begin{tikzpicture}[scale=0.6,>=Latex,
    flow/.style={->, shorten <=2pt, shorten >=2pt}]
    \tikzmath{
      \a = 2;
    }
    \coordinate (v1) at (0, 0) node at (v1) [label=above:$v_{1}$] {};
    \coordinate (v2) at (\a, 0) node at (v2) [label=above:$v_{2}$] {};
    \coordinate (v3) at (0, -\a) node at (v3) [label=below:$v_{3}$] {};
    \coordinate (v4) at (\a, -\a) node at (v4) [label=below left:$v_{4}$] {};
    \coordinate (v5) at (\a, -2*\a)  node at (v5) [label=left:$v_{5}$] {};
    \coordinate (v6) at (2*\a,-2*\a) node at (v6) [label=right:$v_{6}$] {};
    \coordinate (v7) at (\a,-3*\a) node at (v7) [label=left:$v_{7}$] {};
    \coordinate (v8) at (2*\a,-3*\a) node at (v8) [label=right:$v_{8}$] {};
    \coordinate (v9) at (2*\a,0) node at (v9) [label=above:$v_{9}$] {};
    \coordinate (v10) at (3*\a,0) node at (v10) [label=above:$v_{10}$] {};
    \coordinate (v11) at (2*\a,-\a) node at (v11) [label=left:$v_{11}$] {};
    \coordinate (v12) at (3*\a,-\a) node at (v12) [label=below:$v_{12}$] {};
    \draw (v1)--(v2);
    \draw (v1)--(v3);
    \draw (v1)--(v4);
    \draw (v2)--(v4);
    \draw (v2)--(v3);
    \draw (v2)--(v9);
    \draw (v3)--(v4);
    \draw (v4)--(v5);
    \draw (v5)--(v6);
    \draw (v5)--(v7);
    \draw (v5)--(v8);
    \draw (v6)--(v7);
    \draw (v6)--(v8);
    \draw (v7)--(v8);
    \draw (v6)--(v11);
    \draw (v11)--(v9);
    \draw (v11)--(v12);
    \draw (v11)--(v10);
    \draw (v10)--(v12);
    \draw (v6)--(v10);
    \draw (v6)--(v12);
    \fill[fill=white, draw] (v1) circle (3pt) (v2) circle (3pt);
    \fill[fill=white, draw] (v3) circle (3pt) (v4) circle (3pt);
    \fill[fill=white, draw] (v5) circle (3pt) (v6) circle (3pt);
    \fill[fill=white, draw] (v7) circle (3pt) (v8) circle (3pt);
    \fill[fill=white, draw] (v9) circle (3pt) (v10) circle (3pt);
    \fill[fill=white, draw] (v11) circle (3pt) (v12) circle (3pt);
  \end{tikzpicture}
  &\begin{tikzpicture}[scale=0.6,>=Latex,
    flow/.style={->, shorten <=2pt, shorten >=2pt}]
    \tikzmath{
      \a = 2;
    }
    \coordinate (v1) at (0, 0) node at (v1) [label=above:$v_{1}$] {};
    \coordinate (v2) at (\a, 0) node at (v2) [label=above:$v_{2}$] {};
    \coordinate (v3) at (0, -\a) node at (v3) [label=below:$v_{3}$] {};
    \coordinate (v4) at (\a, -\a) node at (v4) [label=below:$v_{4}$] {};
    \coordinate (v5) at (\a, -2*\a)  node at (v5) [label=left:$v_{5}$] {};
    \coordinate (v6) at (2*\a,-2*\a) node at (v6) [label=right:$v_{6}$] {};
    \coordinate (v7) at (\a,-3*\a) node at (v7) [label=left:$v_{7}$] {};
    \coordinate (v8) at (2*\a,-3*\a) node at (v8) [label=right:$v_{8}$] {};
    \coordinate (v9) at (2*\a,0) node at (v9) [label=above:$v_{9}$] {};
    \coordinate (v10) at (3*\a,0) node at (v10) [label=above:$v_{10}$] {};
    \coordinate (v11) at (2*\a,-\a) node at (v11) [label=below:$v_{11}$] {};
    \coordinate (v12) at (3*\a,-\a) node at (v12) [label=below:$v_{12}$] {};
    \draw[flow] (v4) to (v2);
    \draw[flow] (v4) to (v3);
    \draw[flow] (v11) to (v9);
    \draw[flow] (v11) to (v12);
    \draw[flow] (v1) to[out=10,in=170] (v2);
    \draw[flow] (v2) to[out=190,in=350] (v1);
    \draw[flow] (v1) to[out=-100,in=100] (v3);
    \draw[flow] (v3) to[out=80,in=-80] (v1);
    \draw[flow] (v9) to[out=10,in=170] (v10);
    \draw[flow] (v10) to[out=190,in=350] (v9);
    \draw[flow] (v10) to[out=-100,in=100] (v12);
    \draw[flow] (v12) to[out=80,in=-80] (v10);
    \draw[flow, dashed] (v2) to (v9);
    \fill[fill=white, draw] (v1) circle (3pt) (v2) circle (3pt);
    \fill[fill=white, draw] (v3) circle (3pt) (v4) circle (3pt);
    \fill[fill=white, draw] (v5) circle (3pt) (v6) circle (3pt);
    \fill[fill=white, draw] (v7) circle (3pt) (v8) circle (3pt);
    \fill[fill=white, draw] (v9) circle (3pt) (v10) circle (3pt);
    \fill[fill=white, draw] (v11) circle (3pt) (v12) circle (3pt);
  \end{tikzpicture}
  \\
  (i) A graph $G_1$
  & (ii) A graph $G_2$
  & (iii) The propagation digraph 
  \end{tabular}
  \caption{A counterexample
    such that the local monotone system
    based on multi-layer $\bm{k}$-cores
    does not satisfy the in-dominating seed property;
    in (iii), the dashed arc is derived from $G_2$, whereas
    solid arcs are from $G_1$ or both graphs.}
  \label{fig:multi-layer_counterexample}
\end{figure}

Multi-layer $\bm{k}$-cores can be extended further; 
an underlying graph $G_i$, $i\in[1,r]$, can be directed.
For a function $\omega:2^V\times V\to\bbR^r$, write
$\omega_S(v)=(\omega^{(i)}_S(v))_{i=1}^r$
for $S\subseteq V$ and $v\in V$.
Then any local monotone function $\omega^{(i)}:2^V\times V\to\bbR$
can be utilized for each $i\in[1,r]$. 

\paragraph{$(k,\ell)$-Cores~\cite{GTV.2013}.}
We introduce $(k,\ell)$-cores as an example of such extensions of
multi-layer $\bm{k}$-cores.
Let $k,\ell\in\bbZ$ be constants. 
For a digraph $D=(V,A)$,
a subset $S\subseteq V$ (or $D[S]$) is a \emph{$(k,\ell)$-core of $D$}
if $D[S]$ is the maximal subgraph such that
the minimum out-degree is no less than $k$
and the minimum in-degree is no less than $\ell$.
The definition of a $(k,\ell)$-core graph is analogous.

Let $\MZ=(\bbZ^2,\leqslant)$ denote a poset.
We define the function $\omega:2^V\times V$ to be
$\omega_S(v)\coloneqq(\deg^+_S(v),\deg^-_S(v))$, $S\subseteq V$, $v\in V$.
Let $\bm{\theta}=(k,\ell)$ and $\MG=\{G,G\}$. 
The system $\MS=\MS(V,\MZ,\omega,\bm{\theta};\MG)$
is local monotone, and a subset $S\subseteq V$ is a
$\bm{\theta}$-solution if and only if it induces
a $(k,\ell)$-core graph.
By \thmref{main_local_system}, we can enumerate MinRSs of $V$ in 
$O((n+m)n)$ time and $O(n+m)$ space.

\subsubsection{Enumeration of All \texorpdfstring{$\theta$}{theta}-Solutions}
Finally, we show how to enumerate all $\theta$-solutions
in a monotone system, utilizing our framework. 

For a ground set $U$ and $\MF\subseteq 2^U$,
a set system $(U,\MF)$ satisfies the \emph{subset-disjoint}
(\emph{SD}) property~\cite{TH.2025}
if, for any $X,X'\in\MF$ with $X'\subsetneq X$
and any MinRS $Y$ of $X$,
either $Y\subseteq X'$ or $Y\cap X'=\emptyset$ holds. 
We define a \emph{MinRS oracle} to be an oracle
that
returns a MinRS $Y$ of $X$ such that $Y\cap Z=\emptyset$
if one exists, where 
$X\in\MF$ and $Z\subseteq X$ are given as a query to the oracle. 


For an enumeration algorithm,
the \emph{delay} is a crucial performance measure~\cite{JYP.1988},
which refers to time taken
(i) between the start of the algorithm and the output of the first solution;
(ii) between the output of any two consecutive solutions; and
(iii) between the output of the last solution and the termination of the algorithm.
An algorithm is said to be \emph{polynomial delay} 
if its delay is bounded by a polynomial 
function of the input size.

\begin{thm}[Tada and Haraguchi~\cite{TH.2025}]
  \label{thm:TH.2025}
  For a set system $(U,\MF)$ with the SD property,
  we can enumerate all subsets in $\MF$ in
  $O(|U|+\Tmin)$ delay, where $\Tmin$ denotes
  the computation time of a MinRS oracle. 
\end{thm}

We can enumerate all $\theta$-solutions
in a monotone system $\MS(V,\MZ,\omega,\theta)$
efficiently by \thmref{TH.2025} and
the following lemma. 

\begin{lem}
  \label{lem:SD}
  Any monotone system satisfies SD property. 
\end{lem}
\begin{proof}
  For any monotone system $\MS(V,\MZ,\omega,\theta)$,
  let $X_1,X_2$ be $\theta$-solutions
  such that $X_1\supsetneq X_2$.
  For any MinRS $Y$ of $X_1$, we show that
  either $Y\cap X_2=\emptyset$ or $Y\subseteq X_2$ holds. 
  Suppose that $Y\cap X_2\ne\emptyset$.
  Let $Y'\coloneqq Y\setminus X_2$, where $Y'\subsetneq Y$.
  We see $X_2\subsetneq X_1\setminus Y'$, and
  for $v\in Y\cap X_2$, we have
  $\theta\leqslant\omega_{X_2}(v)\leqslant\omega_{X_1\setminus Y'}(v)$. 
  For $u\in X_1\setminus Y$, we have
  $\theta\leqslant\omega_{X_1\setminus Y}(u)\leqslant\omega_{X_1\setminus Y'}(u)$. 
  If $Y'\ne\emptyset$ (i.e., $Y\not\subseteq X_2$), then
  $Y'$ would be an RS of $X_1$,
  contradicting the minimality of $Y$. 
\end{proof}

For our purpose, it suffices to design 
an algorithm to generate MinRSs of any $\theta$-solution.
We can utilize \algref{local_minrs}
although it is designed to enumerate MinRSs of the ground set $V$;
for $S\subseteq V$, the restriction of the monotone system to $S$
is also a monotone system. 

For local monotone systems, we have the following corollary
to Theorems~\ref{thm:main_local_system} and~\ref{thm:in-dominating}. 
\begin{cor}
  \label{cor:main_local_system}
  For a local monotone system $\MS=\MS(V,\MZ,\omega,\theta; G)$,
  we can enumerate all $\theta$-solutions in $\MS$
  in $O((n+m)n\tau_\omega)$ delay and $O(n+m+\kappa_\omega)$ space.
  Furthermore, if $\MS$ satisfies the in-dominating seed property,
  then the delay is $O((n+m)\log n\cdot\tau_\omega)$. 
\end{cor}

We also have the following corollary to \thmref{deg-at-least-k}. 
\begin{cor}
  \label{cor:deg-at-least-k}
  Given an undirected graph $G=(V,E)$ and a positive integer $k$,
  we can enumerate all subsets $S\subseteq V$ that
  induce $k$-core graphs in $O((n+m)\log n)$ delay
  and $O(n+m)$ space. 
\end{cor}

\paragraph{Comparison with Boley et al.'s Work~\cite{BHPW.2010}.}
A set system $(U,\MF)$ is \emph{strongly accessible}
if it is \emph{accessible} (i.e., for all $X\in\MF$, there is $e\in X$ such that $X\setminus\{e\}\in\MF$); and for any $X,Y\in\MF$ such that $X\subsetneq Y$,
there is $e\in Y\setminus X$ such that $X\cup\{e\}\in\MF$. 
Given a strongly accessible system $(U,\MF)$
and a closure operator $\rho:\MF\to\MF$,
Boley et al.~\cite{BHPW.2010}
developed an algorithm that enumerates all fixed points
with respect to $\rho$ (i.e., all sets $C\in\MF$ that satisfy $\rho(C)=C$)
in $O(n(\Tmem+\Tclo+n))$ delay,
where $n=|U|$ and $\Tmem$ and $\Tclo$ denote the computation time
for a membership oracle and for computing the closure $\rho(S)$ of
any subset $S\subseteq U$, respectively.

Their framework on a strongly accessible system $(U,\MF)$
can be used to find all $\theta$-solutions in a
monotone system $\MS(V,\MZ,\omega,\theta)$
by setting $U=V$, $\MF=2^V$ and $\rho$ to the closure operator
that we defined in \secref{mmw_prop}. 
For example, for a given undirected graph $G=(V,E)$, let us consider the problem
of enumerating all $S\subseteq V$ such that $G[S]$ is a $k$-core graph. 
We can implement their algorithm so that $\Tmem=O(1)$ and $\Tclo=O(n+m)$ hold,
and hence the delay is $O((n+m)n)$. 
It is notable that \corref{deg-at-least-k} achieves a better delay bound 
than this straightforward application of Boley et al.'s
framework. 

\section{Concluding Remarks}
\label{sec:conc}

In this paper, we proposed a general framework based on
monotone systems to systematically study the enumeration of MinRSs
and $\theta$-solutions, generalizing $k$-cores and their various extensions.
By introducing the propagation digraph and the SR digraph, we characterized the structural properties of MinRSs.
Under this framework, we developed efficient algorithms
for local monotone systems, achieving a time complexity of $O((n+m)\log n \cdot \tau_{\omega})$ for systems satisfying the in-dominating seed property.
For undirected $k$-cores, we proved that this property holds,
leading to an $O((n+m)\log n)$ time algorithm for MinRS enumeration
and an $O((n+m)\log n)$ delay algorithm for $\theta$-solution enumeration.
The latter delay bound is a significant improvement over the $O((n+m)n)$ delay
achieved by Boley et al.'s strong accessibility-based framework~\cite{BHPW.2010}.

We conclude this paper by highlighting several directions for future work.
(1) Our algorithm still has much room for further investigation.
In particular, although we introduced two auxiliary digraphs
and tailored them for our purpose, 
a deeper theoretical understanding of their properties is required.
It would be interesting to explore the applicability
of the algorithm to problems of finding minimal closures.
(2) Characterization of the in-dominating seed property.
We showed that the in-dominating seed property enables a more efficient algorithm.
However, within the scope of this study, the standard $k$-core is
the only model for which this property holds;
it fails to hold for many of its extensions.
Clarifying the underlying factors that cause this difference remains
a key open question.
We are also interested in investigating weaker versions of
the property that still guarantee computational efficiency.
(3) Empirical evaluations on large-scale real-world networks
are needed to demonstrate the practical performance and scalability
of our proposed algorithms compared to baseline approaches.

\clearpage
\bibliographystyle{abbrvnat}
\bibliography{ref} 

\end{document}